\documentclass[aps,pra,longbibliography,superscriptaddress,twocolumn]{revtex4-1}

\pdfoutput=1

\usepackage[english]{babel}

\usepackage[utf8x]{inputenc}
\usepackage[T1]{fontenc}

\usepackage{float}
\usepackage{graphicx}
\usepackage{epsfig,epstopdf}
\usepackage[table]{xcolor}

\PassOptionsToPackage{hyphens}{url}
\usepackage[hyphenbreaks]{breakurl}

\usepackage{times}

\usepackage{amsmath}
\usepackage{bbm}
\usepackage{amsfonts}
\usepackage{amssymb}
\usepackage{amsthm}
\usepackage{mathbbol}
\usepackage{mathtools}

\usepackage{float}
\usepackage{graphicx}
\usepackage{epsfig,epstopdf}
\usepackage{tikz}
\usepackage{pgffor}

\usepackage{enumerate}
\usepackage{enumitem}

\usepackage{thmtools}
\usepackage{thm-restate}

\usepackage{pgfplots}
\usetikzlibrary{pgfplots.groupplots}
\pgfplotsset{compat=1.11}
\usepgfplotslibrary{fillbetween}

\usepackage{tikz}
\usetikzlibrary{
  shapes,
  shapes.geometric,
	trees,
	matrix,
  positioning,
    pgfplots.groupplots,
  }

\definecolor{newblue}{RGB}{40,210,251}
\definecolor{lightgray}{RGB}{170,170,170}
\definecolor{darkyellow}{RGB}{255,210,70}
\definecolor{darkyellow2}{RGB}{251,184,38}
\definecolor{metalblue}{RGB}{78,156,219}
\definecolor{metalblue2}{RGB}{34,52,103}
\definecolor{pink}{RGB}{237,16,118}
\definecolor{pink2}{RGB}{131,28,71}
\definecolor{violet}{HTML}{53257F} 
\definecolor{violet2}{RGB}{61,18,100}
\definecolor{applegreen}{rgb}{0.55, 0.71, 0.0}
\definecolor{applegreen2}{rgb}{0.4, 0.8, 0.0}

\definecolor{DarkGray}{gray}{0.25} 
\definecolor{MidGray}{gray}{0.38} 
\definecolor{NeutralGray}{gray}{0.5}
\definecolor{LightGray}{gray}{0.7}
\definecolor{lightGray}{gray}{0.85}
\definecolor{DarkRed}{rgb}{0.7,0,0}
\definecolor{DarkBlue}{rgb}{0,0,0.5}
\definecolor{SteelBlue}{rgb}{0,0.4,0.6}
\definecolor{Orange}{rgb}{0.7,0.5,0}
\definecolor{Violette}{rgb}{0.5,0,0.5}
\definecolor{Sand}{rgb}{0.84,0.8,0.55}
\definecolor{niceblue}{rgb}{0.33,0.5,0.8}
\definecolor{OliveGreen}{RGB}{0,102,102}
\definecolor{NiceGreen}{RGB}{0,153,72}

\colorlet{tensorcol}{niceblue!70!gray}
\definecolor{changecol}{rgb}{0.7,0,0}

\colorlet{tensorcol1}{darkyellow}
\colorlet{tensorcol1border}{darkyellow2}
\colorlet{tensorcol2}{metalblue}
\colorlet{tensorcol2border}{metalblue2}

\tikzset{
	ol/.style = {remember picture, overlay},
	eq-pic/.style = {inner sep = .5pt,draw, very thick, gray, rounded corners = 2pt},
}

\tikzset{
  edges/.style = {thick, line cap = round, gray},
  connection edge/.style ={very thick, gray},
  vertex/.style={
    circle,inner sep=0pt,minimum size=2mm,
    thick,draw=gray,
    fill = lightgray
    },
  ball/.style ={
      circle,
      radius=0.01cm,
      shading=ball, 
      ball color=blue},
  site/.style = {
    minimum width = 1.3em, 
    minimum height = 0.5\baselineskip,
    rounded corners=0.3mm,
    thick,draw=black!40,
    top color=white,bottom color=black!20
    },
  generaltensor/.style = {    
    rectangle,
    rounded corners = 0.3mm,
    text = white,
    fill = tensorcol,
    draw
    },
  semicircle/.style = {    
    semicircle,
    rounded corners = 0.3mm,
    text = white,
    fill = tensorcol,
    draw
    },
      tensorbox/.style={
    generaltensor,
    inner sep = 2pt, 
    minimum height = 1.2\baselineskip,
    minimum width = 1.2\baselineskip,
    fill = tensorcol2,
    draw = tensorcol2border
    },
  tensorleg/.style={
    very thick,black!80
    },
  channel/.style = {
    generaltensor,
    minimum height=0.5\baselineskip,
    minimum width=1.5cm
  }, 
1qubit/.style={
    generaltensor,
    inner sep = 2pt, 
    minimum height = 1.2\baselineskip,
    minimum width = 1.2\baselineskip,
    fill = metalblue,
    draw =tensorcol2border,
    thick
    },
1qubitg/.style={
	generaltensor,
	inner sep = 2pt, 
	minimum height = 1.2\baselineskip,
	minimum width = 1.2\baselineskip,
	fill = applegreen,
	draw =DarkGray,
	thick
},
1qubitp/.style={
	generaltensor,
	inner sep = 2pt, 
	minimum height = 1.2\baselineskip,
	minimum width = 1.2\baselineskip,
	fill = pink,
	draw =pink2,
	thick
},
1qubitgray/.style={
	generaltensor,
	inner sep = 2pt, 
	minimum height = 1.2\baselineskip,
	minimum width = 1.2\baselineskip,
	fill = lightray,
	draw =tensorcol2border,
	thick
},
2qubits/.style={
    generaltensor,
    inner sep = 2pt, 
    minimum height = 2.8\baselineskip,
    minimum width = 1.2\baselineskip,
    fill = tensorcol2,
    draw =tensorcol2border,
    thick
    },
2qubitsgray/.style={
	generaltensor,
	inner sep = 2pt, 
	minimum height = 2.8\baselineskip,
	minimum width = 1.2\baselineskip,
	fill = lightgray,
	draw =tensorcol2border,
	thick
},
3qubits/.style={
	generaltensor,
	inner sep = 2pt, 
	minimum height = 4.4\baselineskip,
	minimum width = 1.2\baselineskip,
	fill = tensorcol2,
	draw =tensorcol2border,
	thick
},
3qubitsgray/.style={
	generaltensor,
	inner sep = 2pt, 
	minimum height = 4.4\baselineskip,
	minimum width = 1.2\baselineskip,
	fill = lightgray,
	draw =tensorcol2border,
	thick
},
measure/.style={
    generaltensor,
    inner sep = 2pt, 
    minimum height = 1.2\baselineskip,
    minimum width = 1.2\baselineskip,
    fill = lightgray,
    draw =gray ,
    thick
    }
  }

\PassOptionsToPackage{pdftex,hyperfootnotes=false,pdfpagelabels}{hyperref}
\usepackage{hyperref}
\hypersetup{%
    colorlinks=true, linktocpage=true, pdfstartpage=3, pdfstartview=FitV,%
    breaklinks=true, pdfpagemode=UseNone, pageanchor=true, pdfpagemode=UseOutlines,%
    plainpages=false, bookmarksnumbered, bookmarksopen=true, bookmarksopenlevel=1,%
    hypertexnames=true, pdfhighlight=/O,
    urlcolor=blue, linkcolor=blue, citecolor=black, 
}   

\usepackage{cleveref}

\definecolor{dominik}{rgb}{0.4,.0,0.6}

\makeatletter 
\hypersetup{pdftitle = {},
	     pdfauthor = {},
	     pdfsubject = {Quantum computation, quantum information},
	     pdfkeywords = {Quantum supremacy,
	     	quantum speedup,
	     	universal random circuits, 
	     	commuting circuits, 
	     	diagonal unitaries
		     }
	    }
\makeatother

\newtheorem{theorem}{Theorem}

\newtheorem{definition}[theorem]{Definition}
\newtheorem{lemma}[theorem]{Lemma}

\DeclareMathOperator{\poly}{poly}

\DeclareUnicodeCharacter{2009}{\,}

\newcommand{\ket}[1]{\vert{#1}\rangle}
\newcommand{\bra}[1]{\langle{#1}\vert}

\newcommand{\id}{\mathbbm{1}}

\newcommand{\ii}{\mathrm{i}}

\definecolor{christian}{rgb}{0,.4,1}

\definecolor{jens}{rgb}{0,.8,.5}

\definecolor{juan}{rgb}{.7,.1,0}

\definecolor{dominik}{rgb}{0.4,.0,0.6}

\definecolor{paul}{rgb}{0.4,.5,0.6}

\definecolor{tm}{rgb}{0,1,0}

\newcommand{\fu}{Dahlem Center for Complex Quantum Systems, Freie Universit{\"a}t Berlin, 14195 Berlin, Germany}

\newcommand{\mfu}{Department of Mathematics and Computer Science, Freie Universit{\"a}t Berlin, 14195 Berlin, Germany}
\newcommand{\hzb}{Helmholtz-Zentrum Berlin f{\"u}r Materialien und Energie, 14109 Berlin, Germany}

\begin{document}

\title{Closing gaps of a quantum advantage with short-time Hamiltonian dynamics}

\author{J.\ Haferkamp}

\affiliation{\fu}

\author{D.\ Hangleiter} 
\affiliation{\fu}

\author{A.\ Bouland}
\affiliation{Department of Electrical Engineering and Computer Sciences, University of California, Berkeley}

\author{B.\ Fefferman}
\affiliation{Department of Computer Science, The University of Chicago}

\author{J.\ Eisert}
\affiliation{\fu}
\affiliation{\hzb}
\affiliation{\mfu}
\author{J.\ Bermejo-Vega}
\affiliation{\fu}

\begin{abstract} 
Demonstrating a quantum computational speedup is a crucial milestone for near-term quantum technology.
Recently, quantum simulation architectures have been proposed that have the potential to show such a quantum advantage, based on
commonly made assumptions. The key challenge in the theoretical analysis of this scheme -- as of other comparable schemes such as boson sampling -- is to lessen the assumptions and close the theoretical loopholes, replacing them by rigorous arguments. In this work, we prove two open conjectures for these architectures for Hamiltonian quantum simulators: anticoncentration of the generated probability distributions and average-case hardness of exactly evaluating those probabilities. 
The latter is proven building upon recently developed techniques for random circuit sampling. For the former, we develop new techniques that exploit the insight that approximate $2$-designs for the unitary group admit anticoncentration. We prove that the 2D translation-invariant, constant depth architectures of quantum simulation form approximate $2$-designs in a specific sense, thus obtaining a significantly stronger result. 
Our work provides the strongest evidence to date that Hamiltonian quantum simulation architectures are classically intractable.  
\end{abstract}

\maketitle

Quantum computers and simulators are expected to greatly outperform classical devices when solving certain tasks.
 Famous examples of such tasks include the factorization of integers~\cite{shor_polynomial-time_1999} and the simulation of Hamiltonian dynamics \cite{feynman_simulating_1982,lloyd_universal_1996,BlochSimulation}. 
While the importance of these results can hardly be overemphasized, the realization of devices capable of outperforming classical computers for practical problems appears to be far out of reach for current technology \cite{Reiher7555,campbell_applying_2018,litinski_game_2019,gidney_how_2019}. 
A key milestone in the development of quantum computers and simulators is therefore to assess the possibility of performing computations that cannot be efficiently reproduced by a classical computer, a state of affairs referred to as a quantum
advantage or ``quantum supremacy''.
Besides being a technological breakthrough, 
such an experiment  
can be regarded as the first experimental violation of the Extended Church-Turing thesis, and will be a watershed moment in the history of computation. 

In order to conclusively demonstrate the superior computational power of quantum devices we must hold ourselves to a particularly high standard of evidence.  
While several examples of large-scale experimental quantum simulators that outperform certain classical algorithms have been reported~\cite{Trotzky,Schreiber-pnas-2015,choi_exploring_2016,bernien_probing_2017,Monroe}, 
to have high confidence that these devices are providing bona fide quantum speedups, 
we must give evidence that \emph{no classical algorithm} will ever be able to solve this problem efficiently. 
This has been advanced by recent work providing evidence for the computational hardness of certain sampling tasks that are both robust against some errors and feasible on near-term quantum devices~\cite{aaronson_bosonsampling_2010,Bremner}.  
Indeed, these hardness-of-sampling results are widely viewed as the most promising avenue to achieving a provable quantum advantage in the near future.

However, several key open problems are outstanding for this approach.
First, in the near future only imperfect and small universal quantum devices are becoming available in laboratories around the world~\cite{boixo2016characterizing,wang2017high}.  
A key open question is hence to identify a task that is both feasible on the large-scale quantum hardware that is available today \emph{and} for which complexity-theoretic evidence for hardness can be provided. 
Second, it is crucial that the achievement of an advantage is verified~\cite{harrow_quantum_2017}, a daunting task~\cite{hangleiter_sample_2019,gogolin_boson-sampling_2013,aaronson_bosonsampling_2013,CertificationReview} 
given its computational hardness and the sheer size of the sample space. 
Finally, the hardness results rely on unproven albeit plausible conjectures beyond standard complexity-theoretic assumptions.
Answering these questions requires building new tools at the interface 
of quantum many-body physics and computational complexity theory.

Quantum advantage schemes for quantum simulators that involve the constant-time evolution of translation-invariant Ising Hamiltonians have been proposed in Refs.~\cite{gao_quantum_2017,bermejo-vega_architectures_2018}.
Those architectures show a provable quantum advantage under similar assumptions as in Refs.~\cite{aaronson_bosonsampling_2010,bremner_average-case_2016} for large-scale quantum simulators that are available and outperform known classical algorithms already today~\cite{Trotzky,Schreiber-pnas-2015,choi_exploring_2016,Monroe}.
At the same time, they admit an efficient and rigorous certification protocol that only requires partial trust in single-qubit measurements~\cite{hangleiter_direct_2017,bermejo-vega_architectures_2018}.
This constitutes a key step towards facilitating the large scale experimental realization of quantum advantages and closing the certification loophole. 
This approach is entirely measurement-based and hence different from gate-based proposals~\cite{aaronson_bosonsampling_2010,Bremner,boixo_characterizing_2018,bouland_quantum_2018}.

The central open problem in the complexity theoretic argument for hardness of all such sampling schemes revolves around its robustness to noise~\cite{aaronson_bosonsampling_2010,Bremner}. 
This argument builds upon ideas from Ref.~\cite{terhal_adaptive_2004} that shows the hardness of \emph{exact sampling} for certain models based on commonly believed complexity assumptions (i.e., a generalization of $\mathsf{P} \neq \mathsf{NP}$ called non-collapse of the polynomial hierarchy) using a technique called Stockmeyer's algorithm \cite{stockmeyer_approximation_1985}. 
To make these results noise-robust, the key idea, developed by \cite{aaronson_bosonsampling_2010}, is to make use of average, rather than worst-case complexity.
In particular, they showed that noise robust sampling hardness would follow if one could show that it is very hard to \emph{approximate} the output probabilities of \emph{most} randomly chosen quantum circuits.

Proving this key conjecture called \emph{approximate average-case hardness} has remained elusive for all known practical schemes that are amenable to the Stockmeyer proof strategy. 
Aaronson and Arkhipov have also observed, however, that evidence for approximate average-case hardness can be provided using certain properties of the sampled distribution:
First, \emph{exact average-case hardness} constitutes a necessary criterion for the approximate version thereof.
Second, the so-called \emph{anticoncentration} property reduces the notion of approximation that is necessary for the hardness proof to a more plausible one that involves only relative errors. 
Indeed, both of these loopholes have recently been closed for the prominent universal circuit sampling proposal~\cite{bouland_quantum_2018,brandao_local_2016,hangleiter_anticoncentration_2018,harrow_approximate_2018}.

In this work, we close both loopholes simultaneously for the simple quantum simulation architecture on a square lattice of Ref.~\cite{bermejo-vega_architectures_2018} -- thus bringing it up to the highest standard to date for evidence for computational intractability. 

First, we prove \emph{anticoncentration} for this model. 
In fact, our main contribution is to establish an even stronger property than anticoncentration, namely, that the effective circuits generated by the architectures mimic Haar-randomness up to second moments -- surprisingly -- already on square ($n \times \mathcal O(n)$) lattices.
In precise terms, we prove that these circuits form an approximate $2$-design. 
\begin{restatable}[Approximate $2$-design]{thm}{design}\label{theorem:t-design}
	Consider the architectures of quantum simulation with local rotation angles chosen uniformly from $[0, 2\pi)$ on an $n\times m$ lattice with $m\in\mathcal{O}\left(4n+\log\left(1/\varepsilon\right)\right)$.
	When measuring the first $m-1$ columns in the $X$-basis, the effective unitary acting on the last column forms a relative $\varepsilon$-approximate unitary $2$-design.
\end{restatable}
\addtocounter{theorem}{1}

Numerical evidence provided in Ref.~\cite{bermejo-vega_architectures_2018} suggests that anticoncentration happens already for $n \times n$ lattices. 
And in fact, the emergence of relatively $\varepsilon$-approximate $2$-designs is a much more powerful result than mere anticoncentration. 
First, observe that the $2$-design property directly implies anticoncentration \cite{hangleiter_anticoncentration_2018,bouland2018CCCofCCCs,mann_complexity_2017,harrow_approximate_2018}. 
Second, we note that generating the moments of the Haar measure is considered even stronger evidence for hardness of classical simulation than mere anticoncentration~\cite{boixo2016characterizing}. 
What is more, $2$-designs in fact find a number of applications such as decoupling~\cite{szehr_decoupling_2013,hirche_decoupling_2013,dupuis_decoupling_2014} and 
randomized benchmarking~\cite{emerson_benchmarking_2009} and robust quantum gate tomography~\cite{PhysRevLett.121.170502}.

But already rigorously establishing anticoncentration is a difficult endeavour, in our case particularly so due to the low depth involved.
Indeed, for the case of random circuit sampling anticoncentration holds at depth $O(\sqrt{n})$ on a $\sqrt{n} \times \sqrt{n} $ 2D grid \cite{harrow_approximate_2018}, and at depth $O(n)$ in 1D \cite{brandao_efficient_2016,hangleiter_anticoncentration_2018}, but it is not expected for constant depth~\cite{boixo_characterizing_2018,bremner_achieving_2017}.
With our work, we show anticoncentration at much lower -- constant -- depth but in a different model of random circuits which obey a form of translation invariance. 
Our result implies the first non-trivial anticoncentration bound for constant-time, translation-invariant dynamics on a square lattice, going significantly beyond direct measurement-based embeddings~\cite{miller_quantum_2017,hangleiter_anticoncentration_2018,gao_quantum_2017,mezher_efficient_2018,mezher2019}.
It can also be seen as an analytical proof of the two-design property, which was numerically explored in a similar measurement-based scheme~\cite{brown_quantum_2008_1}.

Our proof of the $2$-design property is inspired by and significantly develops further a recent result of \citet{brandao_local_2016} that 
shows that random universal circuits form an approximate $t$-design. 
Explicitly, we exploit the connection to gaps of frustration-free Hamiltonians~\cite{brown_convergence_2010,brandao_local_2016}. 
We follow the general strategy of \citet{brandao_local_2016}, but every individual step of the proof requires new methods, which might be of independent interest.
In particular, we prove that the effective circuits generated by the translation-invariant time evolution are computationally universal -- a non-trivial task given that those circuits are \emph{not locally universal}. 
We then exploit a recent generalization of the detectability lemma \cite{anshu_detectability_2016} and the famous martingale method pioneered by~\citet{nachtergaele_gap_1994} to lower-bound the spectral gap of the detectability Hamiltonian.

As our second main contribution, we prove average-case hardness for \emph{exactly evaluating} the output probabilities of the architectures. 
To do so, we extend a recent result of \citet{bouland_quantum_2018} showing exact average-case hardness of universal circuit sampling \cite{boixo2016characterizing} to the translation invariant case. 
Informally, we obtain the following result: 
\begin{theorem}[Average-case hardness]\label{theorem:informalaverage-case}
 It is \#P-hard to exactly compute any $3/4+1/\poly(N)$ fraction of the output probabilities of the architectures of quantum simulation.
\end{theorem}
Our work also demonstrates that these average-case hardness methods are applicable to many other sampling architectures, such as continuous forms of IQP circuits~\cite{Bremner} and other measurement-based schemes \cite{gao_quantum_2017,miller_quantum_2017,mezher2019}.

\paragraph*{Architectures of quantum simulation showing a quantum speedup.}
Our new analysis 
builds on the proposal for a quantum speedup from Ref.~\cite{bermejo-vega_architectures_2018}, which we recall here.
It is a scheme that is much reminiscent of a quench-type quantum simulation, involving the 
evolution of a product state under a nearest-neighbour Hamiltonian for a constant time.
Here, we define a slightly modified protocol, as the angles are drawn Haar-randomly from $S^1$ and not discretely as in Ref.~\cite{bermejo-vega_architectures_2018}:
\begin{itemize}
	\item \textbf{Preparation}: Arrange $N \coloneqq nm$ qubits on an $n$-row $m$-column lattice $\mathcal{L}$, with vertices $V$ and edges $E$. 
	Prepare the product state vector
	\begin{equation}\label{eq:preparation}
	|\psi_{\beta}\rangle=\bigotimes_{i=1}^N\left(|0\rangle+e^{\ii\beta_i}|1\rangle\right),\, \beta\in[0,2\pi)^N,
	\end{equation}
	for $\beta$ chosen randomly from the Haar measure on $(S^1)^{\times N}$. 
	$S^1$ denotes the circle $[0,2\pi]/\sim$, where $\sim$ identifies $0$ and $2\pi$.
	\item \textbf{Time evolution}: Let the system evolve for constant time $\tau=1$ under a nearest-neighbour and translation-invariant Ising Hamiltonian
	\begin{equation}
	H \coloneqq \sum_{(i,j)\in E} J_{i,j} Z_iZ_j-\sum_{i\in V}h_iZ_i,
	\end{equation}
	with constants $J_{i,j}$ and $h_i$ chosen to implement a unitary $e^{\ii H}$.
	This amounts to a constant depth circuit.
	\item \textbf{Measurement}: Measure all qubits in the $X$ basis.
\end{itemize}

This protocol can be translated to the setting of deep quantum circuits via measurement based quantum computing. 
In particular, it can be proven similarly to Ref.~\cite{bermejo-vega_architectures_2018} that the above architecture is equivalent to a circuit with randomly drawn gates acting on the last column of $n$ qubits.
We can express this random circuit with Haar-randomly drawn angles $\tilde{\beta}_i^j$ as
\begin{equation}\label{eq:circuitpicture}
	U_\beta^m =E\left(\prod_{i=1}^n e^{\ii\tilde{\beta}^{m-1}_iZ_i}\right)E...E\left(\prod_{i=1}^n e^{\ii\tilde{\beta}^1_i Z_i}\right),
\end{equation}
 with the global entangling unitary
\begin{equation}\label{eq:def-E}
E \coloneqq \left(\prod_{i=n}^n H_i\right)\left(\prod_{i=1}^{n-1}CZ_{i,i+1}\right),
\end{equation}
where $H_i$ denotes the Hadamard gate acting on the $i$th qubit and $CZ$ denotes the controlled $Z$ gate.
Based on the conjectures of approximate average-case hardness and anticoncentration, the protocol was shown to yield a superpolynomial speed up with high probability~\cite{bermejo-vega_architectures_2018} using the techniques of Ref.~\cite{aaronson_bosonsampling_2010}.

\paragraph*{The 2-design property and anticoncentration.}
We now prove anticoncentration for the architecture in Ref.~\cite{bermejo-vega_architectures_2018} with 
continuous angle choices.
For a proof in full technical detail, we refer to Appendix~\ref{appendix:proof-2-design}.
Consider a distribution $v$ on the unitary group $\mathbb{U}(2^n)$ acting on $n$ qubits and the corresponding output probabilities $|\langle x|U|0\rangle|^2$ for obtaining $x\in\{0,1\}^n$. 
We say that $v$ \textit{anticoncentrates} if there exist constants $\alpha, \beta>0$ such that for any fixed $x\in\{0,1\}^n$ 
\begin{equation}
	\Pr_{U\sim v}\left(|\langle x|U|0\rangle|^2\geq \frac{\alpha}{2^n}\right)\geq\beta .
\end{equation}
Instead of proving this directly, we show a stronger property: 
the architectures of quantum simulation define relative $\varepsilon$-approximate $2$-designs in the sense that the random circuit $\tilde U_\beta^m$ relative $\varepsilon$-approximate unitary $2$-design. 
That is, the deep random circuit $U_\beta^m$ approximates the first and second moments of the Haar measure up to a relative error $\varepsilon$: 
\begin{definition}[Relative $\varepsilon$-approximate unitary $t$-designs]
	Let $v$ be a distribution on $\mathbb{U}(N)$. Then, $v$ is an $\varepsilon$-approximate $t$-design if
	\begin{equation}
	(1-\varepsilon)\Delta_{\mu_{\rm Haar},t}\preccurlyeq \Delta_{v,t}\preccurlyeq (1+\varepsilon)\Delta_{\mu_{\rm Haar},t},
	\end{equation}
	where the superoperator $\Delta_{v,t}$ is defined via
	\begin{equation}
	\Delta_{v,t}(\rho) \coloneqq \int_{\mathbb{U}(N)}U^{\otimes t}\rho \left(U^{\dagger}\right)^{\otimes t}\mathrm{d}v(U),
	\end{equation}
	and $A\preccurlyeq B$ if and only if $B-A$ is completely positive.
\end{definition}
It has been observed that this property together with the \textit{Paley-Zygmund inequality} yields anticoncentration \cite{hangleiter_anticoncentration_2018,bouland2018CCCofCCCs,mann_complexity_2017,harrow_approximate_2018}:
\begin{lemma}[\cite{hangleiter_anticoncentration_2018}]
Let $v$ be an $\varepsilon$-approximate unitary $2$-design on $\mathbb{U}(2^n)$. Then, $v$ anticoncentrates in the sense that for $0\leq \alpha\leq 1$ and for all $x\in\{0,1\}^n$ we have
\begin{equation}
\Pr_{U\sim v}\left(|\langle x|U|0\rangle|^2>\frac{\alpha(1-\varepsilon)}{N}\right)\geq \frac{(1-\alpha)^2(1-\varepsilon)^2}{2(1+\varepsilon)}.
\end{equation}
\end{lemma}
Thus, if the deep circuit family $\{U_\beta^m\}_\beta$ forms an $\varepsilon$-approximate unitary $2$-designs the remaining quantum state on the last column anticoncentrates. 
But this already implies that the full output distribution $p_\beta(x) \coloneqq |\langle x|\exp(\ii H) |\psi_{\beta}\rangle|^2$ anticoncentrates by a property of measurement-based quantum computation:
Let $x_L \in \{ 0,1\}^{n (m-1)}$ be a string of outcomes obtained from measuring the first $m-1$ columns, and $x_R \in \{0,1\}^n$ a string of outcomes obtained from measuring the last column. 
Then $p(x) \equiv p(x_L, x_R) = p(x_R|x_L) p(x_L) = p(x_R|x_L) /2^{n(m-1)}$. 
Hence, $p(x)$ anticoncentrates, if $p(x_R|x_L)$ does. 
But this is proven by the $2$-design property of the depth-$(m-1)$ circuit  $U _\beta^m$.

Notice that the \emph{relative-error} notion of approximate $2$-designs which we use here is distinct from an \emph{additive-error} definition of $t$-designs, which is much weaker and in particular would not suffice to prove anticoncentration.

\begin{proof}[Proof of Theorem~\ref{theorem:t-design} (outline)]
	
Proving Theorem~\ref{theorem:t-design} amounts to proving the relative $\varepsilon$-approximate $2$-design property of the depth-$m$ random circuit family $\{U^m_\beta\}_\beta$.
In the proof, we follow a strategy for showing this type of result pioneered in Refs.~\cite{harrow_random_2009,brandao_local_2016}:
the key idea behind this strategy is to successively reduce the $2$-design property of the full circuit to a simpler property, namely the spectral gap of a certain frustration-free Hamiltonian~\cite{brown_convergence_2010}.

More precisely, we proceed in three steps.
In the first step, we reduce the $2$-design property to a so-called $2$-copy tensor product expander (TPE) property defined as an upper bound on the quantity $g(v,2)$, which is defined as follows. 
Given a distribution $v$ on the unitaries, let 
\begin{equation*}
 g(v,2)\coloneqq\left|\left| \ \int \limits_{\mathbb{U}(2^n)}U^{\otimes 2,2}\mathrm{d}v(U)-\int\limits_{\mathbb{U}(2^n)}U^{\otimes 2,2}\mathrm{d}\mu_{\rm Haar}(U)\right|\right|_{\infty}
\end{equation*}
with $U^{\otimes 2,2} \coloneqq U\otimes U\otimes U^*\otimes U^*$.
Tensor-product expanders are similar to approximate $2$-designs in the sense that if 
\begin{equation}
g(v,2)\leq 2^{-4n}\varepsilon,
\end{equation}
then $v$ is a relative $\varepsilon$-approximate $2$-design~\cite{brandao_local_2016}.
$g(v,2)$ has the convenient property that $g(v^{*k},2)\leq g(v,2)^k$, where $v^{*k}$ is the $k$-fold convolution.
The $k$-fold convolution corresponds to concatenations of the form $U_1 \cdots U_k$ with each $U_i$ drawn from $v$.
Since the circuit $U_\beta^m$ is a concatenation of $m-1$ unitaries that are distributed according to the same measure, this property allows to reduce the proof to proving the TPE property for one or more steps of the circuit. 
 
To prove the $2$-design property for~\eqref{eq:circuitpicture} an obvious but ultimately not fruitful approach is thus to choose $v$ to be one layer $E\left(\prod_{i=1}^{n}e^{\ii\tilde{\beta}_iZ_i}\right)$ in the circuit~\eqref{eq:circuitpicture}.
Instead, we choose three such layers and rewrite it in the form $E^2UE$ with fixed global unitaries $E$ and $E^2$. We obtain
 \begin{align}
 U=&\left(e^{\mathrm{i}\varphi^{ZXZ_1} ZXZ_1}... e^{\mathrm{i}\varphi^{ZXZ}_n ZXZ_n}\right)\nonumber\\
 &\left(e^{\mathrm{i}\varphi^Z_1 Z_1}...e^{\mathrm{i}\varphi^Z_n Z_n}\right)\left(e^{\mathrm{i}\varphi^X_1 X_1}... e^{\mathrm{i}\varphi^X_n X_n}\right)\label{eq:unitaries}
 \end{align}
 for Haar-randomly drawn $\varphi^X_i$, $\varphi^Z_i$ and $\varphi^{ZXZ}_i$, with the notation $ZXZ_i=Z_{i-1}\otimes X_i\otimes Z_{i+1}$ for $2\leq n-1$, $ZXZ_n=Z_{n-1}\otimes X$ and $ZXZ_1=Z_1\otimes X_2$. 
 The correct intuition here is that fixed unitaries do 
 not alter the degree of randomness and that we can simply remove $E$ and $E^2$.
 Indeed, we can prove for some any distribution $v$ and general unitaries $V$ and $W$ that $g\left(VvW,2\right)= g(v,2)$.  
 It thus suffices to bound $g(v_n,2)$, where $v_n$ denotes the distribution that $U$ is drawn from.

 However, for such a bound to be feasible it is crucial that the distribution $v$ contains a universal gate set.
 We prove this using the fact that $e^{\ii\alpha Z}$-gates and $e^{\ii\alpha X}$-gates are dense in the single qubit unitaries.
 Furthermore, any additional two-qubit entangling gates suffices to obtain full universality~\cite{brylinski_universal_2001,bremner_universal_2002}.
 We have entangling two-qubit gates on the boundary and entangling three-qubit gates in the bulk. 
 Using the the boundary unitaries, we can propagate the universality into the bulk.
 
In the next step, we reduce the tensor-product expander property of $U$, that is, an upper bound on $g(v_n,2)$ to a spectral gap of a certain frustration-free Hamiltonian. 
To do so, we apply a generalized version~\cite{anshu_detectability_2016} of the so-called \textit{detectability lemma}~\cite{aharonov_detectability_2008}, which yields the bound 
 \begin{equation}
 g(v_{n},2)\leq \frac{1}{\sqrt{\frac{\Delta(H_n)}{9}+1}},
 \end{equation}
 where $\Delta(H_n)$ denotes the \textit{spectral gap} of the local Hamiltonian $H_n$, i.e., the difference 
 between its first and second eigenvalue.
This Hamiltonian is defined as
 \begin{equation}
 H_n \coloneqq \sum_{i=1}^n\left(\mathbb{1}-P_i^{X}\right)+\sum_{i=1}^n\left(\mathbb{1}-P_i^Z\right)+\sum_{i=1}^n\left(\mathbb{1}-P_i^{ZXZ}\right),
 \end{equation}
 with local orthogonal projectors 
 \begin{align*}
 P^X_i& \coloneqq \frac{1}{2\pi}\int \left(e^{\mathrm{i}\varphi^X_i X_i}\right)^{\otimes 2,2}\mathrm{d}\varphi^X_i,\\
 P^Z_i& \coloneqq \frac{1}{2\pi}\int \left(e^{\mathrm{i}\varphi^Z_i Z_i}\right)^{\otimes 2,2}\mathrm{d}\varphi^X_i,\\
 P^{ZXZ}_i& \coloneqq \frac{1}{2\pi}\int \left(e^{\mathrm{i}\varphi^{ZXZ}_i ZXZ_i}\right)^{\otimes 2,2}\mathrm{d}\varphi^{ZXZ}_i.
 \end{align*}
 
In the last step of the proof, we find a lower bound to the spectral gap of $H_n$. 
Proving that a local Hamiltonian has a spectral gap in the thermodynamic limit is in general a highly non-trivial task.
In fact, deciding whether a general Hamiltonian is gapped in 
the thermodynamic limit is known to be undecidable~\cite{cubitt_gap_2015,bausch_gap_2018}.
To prove a lower bound to the spectral gap of $H_n$ we exploit that it is \textit{frustration-free} so that the global ground states simultaneously minimize all local Hamiltonian terms.
This property can be used as leverage to tackle the problem of lower bounding the spectral gap. 
In particular, we apply the \textit{Nachtergaele bound}~\cite{nachtergaele_gap_1994}, sometimes called \textit{martingale trick}.
This method requires frustration-freeness, finite range of interactions and a third condition concerning the overlap of ground state projectors $G_{[m,n]}$ of the local Hamiltonians $\id^{\otimes m-n}\otimes H_{[m,n]}\otimes \id^{\otimes N-n+1}$.

 The ground space for intervals containing the boundary terms corresponds to a well studied problem in the context of the Schur-Weyl duality, namely, characterizing the set of matrices that commute with all unitaries of the form $U^{\otimes 2}$.
These are precisely the standard representations of the symmetric group $S_2$, thus the corresponding ground space is $2$-dimensional.
 The bulk Hamiltonian admits a larger ground space but -- perhaps surprisingly -- turns out to be explicitly computable and always of dimension $3$.
 Using (non-orthonormal) basis states for these ground spaces, we can construct approximations to the ground state projectors.
 This allows us to verify the third condition for $l=6$. 
 Then, the Nachtergaele bound can be applied and yields
 \begin{equation}
 \Delta(H_n)\geq \frac{\Delta\left(H^B_7\right)}{32}\quad\text{for all}\quad n\geq 8,
 \end{equation}
 where $H^B_n$ is the Hamiltonian $H_n$ without the two-local boundary terms.
 As $\Delta(H^B_7)>0$ is a constant, we have proven that the circuits $U_\beta^m$ defined in Eq.~\eqref{eq:circuitpicture} is a relative $\varepsilon$-approximate $2$-design in depth $m \in \mathcal{O} (4n+\log(1/\varepsilon))$.
\end{proof}

 \paragraph*{Average-case hardness.}

In this section, we prove average-case hardness of calculating the output probabilities of the architecture using polynomial interpolation techniques. 
Our argument follows the proof strategy developed in Ref.~\cite{bouland_quantum_2018}. 
Specifically, we show that a machine $\mathcal{O}$ that computes a certain fraction of all instances can be used to construct a machine $\mathcal{O}$ that solves all instances in random polynomial time, a property known as \textit{random self-reducibility}.
But solving all instances is known to be \#P hard, which implies that solving that fraction of instances is just as hard. 
We now sketch the proof of this statement. 
For a detailed proof and a technical statement, we refer to Appendix~\ref{appendix:proof_average}.
Moreover, we provide a generalization of this statement in Appendix~\ref{appendix:generalized-average}, which contains the statement in Ref.~\cite{bouland_quantum_2018} as a special case and moreover shows average-case complexity for commuting quantum circuits (IQP circuits)~\cite{Bremner}.

For the proof it suffices to consider the probability of a single output probability, which we choose w.l.o.g.\ to be $p_{0,\beta} = |\langle 0|\exp(\ii H) |\psi_{\beta}\rangle|^2$. 
This is a consequence of the so-called \textit{hiding property}, which refers to the fact that we can hide all output strings in the circuit without changing the probabilities with which the circuits are drawn. 
To see this, note that we can write any state $| 1 \rangle$ of a qubit in the output state as $X|0\rangle$.
But the operator $X$ can be propagated through the circuit~\eqref{eq:circuitpicture} to meet a Hadamard gate, where it becomes the operator $Z$. 
Together with a gate $\exp(\ii\alpha Z)$ it then forms the gate $\ii \exp(\ii(\alpha+\pi/2) Z)$.
By translation-invariance of the Haar measure, the angle $\pi/2$ does not change the probabilities for the circuit.

We can thus show the worst-to-average reduction for computing the output probability $p_{0,\beta}$. 
This reduction is inspired by a similar reduction due to \citet{lipton_permanent_1991} for the \textit{permanent} of a matrix: 
Consider an instance $p_{0,\beta}$ of our computational problem defined by $\beta=(\beta_{1},...,\beta_{N})$.
Suppose, we draw an instance $\gamma=(\gamma_1,...,\gamma_N)$ from the Haar measure on $(S^1)^{\times N}$. 
The reduction is based on an interpolation between the fixed instance $\beta$ and the randomly drawn instance $\gamma$.
This can be achieved by linear interpolation between the angles: $\eta(\theta) \coloneqq \theta\beta+(1-\theta)\gamma$. 
However, for worst-to-average reduction similar to the permanent to work, we need that  $|\langle 0|\exp(\ii H) |\psi_{\beta}\rangle|^2$ is a polynomial in $\theta$. 
This can be dealt with by truncating the Taylor expansion of the gates: 
In the circuit picture, the above interpolation corresponds to
$\exp(\ii(\theta\beta_i+(1-\theta)\gamma_i)Z)$. Instead we can consider the gate
\begin{equation}
	G_i(\theta)=e^{\ii\gamma_i Z}\left(\sum_{k=0}^K\frac{(\ii\theta(\beta_i-\gamma_i) Z)^k}{k!}\right),
\end{equation}
with $K=\poly(N)$.
This defines a circuit $\tilde{U}(\theta)$ for which the output probability $p_{0,\beta}(\theta)=|\langle 0|\tilde{U}(\theta)|0\rangle|^2$ can be shown to be a polynomial in $\theta$ with degree polynomial in $n$, using a Feynman path integral form.
Furthermore, we can show that $G_i(\theta)$ is drawn from a distribution that is arbitrarily close to the Haar measure.
Similarly to the reduction for the permanent, we can now query the machine $\mathcal{O}$ polynomially many times to recover the polynomial $p_{0,\beta}(\theta)$. 
Using modern techniques for the recovery of polynomials such as the \textit{Berlekamp-Welch algorithm}~\cite{berlekamp_polynomial_1986}, we can bound the probability with which $\mathcal{O}$ needs to be correct.
In particular, we can prove (with the Markov inequality) that if $\mathcal{O}$ solves a fraction of $3/4+1 /\poly(N)$ of the instances drawn from the perturbed Haar measure, then one can solve all instances with a probability of $1/2+1 /\poly(N)$. 
Moreover, by repeating the above procedure polynomially many times and taking a 
majority vote over all trials, this probability can be exponentially amplified.

We note that, strictly speaking, our result does not prove average-case hardness of the Haar distribution on $S^1$ but a close distribution with the property that it takes values outside of the unitary group which are $1/2^{\poly(n)}$ close to the ideal ones. 
\citet{movassagh_efficient_2018} provided a fix for this technical caveat by replacing the polynomial interpolation step with a rational-function interpolation which is based on the $QR$-decomposition.
The polynomial interpolation method has the advantage, however, to allow for $1/2^{\poly'(n)}$ robustness to noise at the cost of reducing the fraction of hard instances to a polynomially small one. 
This level of robustness is crucial to formalize the result in a Turing machine model which has finite precision. 

\paragraph*{Conclusion.}
In this work, we have solved two of the open conjectures that provide the theoretical footing of quantum simulation proposals in Ref.~\cite{bermejo-vega_architectures_2018}
showing a complexity theoretic quantum speedup.
We have, thus, provided a translation-invariant, nearest-neighbour and constant depth proposal for a quantum advantage with the strongest theoretical evidence that is feasible with state-of-the-art methods. 
The conjecture of approximate average-case complexity remains open for all quantum speedup proposals and is the only missing step to a full loophole free complexity-theoretic argument.
These observations render the scheme considered in this work among the most stringent quantum advantage schemes for which the conceptual loopholes are most convincingly closed.

\paragraph*{Acknowledgments.}
We would like to thank Michael Bremner and Ashley Montanaro for comments on the manuscript. 
This work has been supported by the ERC (TAQ), the Templeton Foundation, and
the DFG (EI 519/14-1, EI 519/15-1, CRC 183) and MATH+. 
This work has also received funding from the European Union's Horizon 2020 research and innovation 
programme under grant agreement No 817482 (PASQuanS). A.~B. was supported in part by ARO Grant W911NF-12-1-0541, NSF Grant CCF-1410022, and a Vannevar Bush faculty fellowship. B.~F. acknowledges support from AFOSR YIP number FA9550-18-1-0148.  Portions of this paper are a contribution of NIST, an 
agency of the US government, and is not subject to US copyright.

\bibliographystyle{myapsrev4-1}

%

\onecolumngrid

\appendix

\section{The full Hamiltonian and mapping to effective circuits}
\label{app:effective circuit mapping}
The Ising Hamiltonian for the architectures of quantum simulation reads
\begin{equation}
H \coloneqq \sum_{(i,j)\in E}\frac{\pi}{4}Z_iZ_j-\sum_{i\in V}\frac{\pi}{4}\mathrm{deg}(i)Z_i,
\end{equation}
where $\mathrm{deg}(i)$ is the number of adjacent edges to the vertex $i$, i.e., 
it takes the values $2,3$ or $4$ depending on whether $i$ is located at an edge, the inner boundary or the bulk.
We now turn to showing how the architecture of quantum simulation described in the main text can be mapped to a circuit of the form~\eqref{eq:circuitpicture}.
In more detail, we show that the architectures are equivalent to the following random circuits (compare Fig.~\ref{figure:architecture}):

\begin{definition}[Effective random circuit]\label{def:architecture}
	We define an effective circuit by the following protocol:
	\begin{enumerate}
		\item Prepare each qubit in the state $\ket +$.
		\item For each qubit, draw a phase $\varphi_j\in S^1\stackrel{\sim}{=}[0,2\pi]/\sim$ uniformly at random and apply the diagonal gate $G_j \coloneqq e^{\ii \varphi_j Z}$.
		\item Apply a controlled $Z$ gate $CZ$ to all neighbouring qubits.
		\item Apply a Hadamard gate to each qubit.
		\item Repeat the above $D=\poly(N)$ many times.
		\item Measure in the $Z$ eigenbasis.
	\end{enumerate}
	In general, we refer to the resulting quantum circuits arising from drawing random $G_j$ as \emph{random circuits}.
\end{definition}
The task that we will show to be average-case hard is to sample from the output distribution of this circuit.
In general, we restrict to families of graphs that are uniformly bounded, i.e. the number of edges associated to each vertex is bounded by a constant.

\begin{figure}[H]\label{figure:architecture}
	\centering
	\includegraphics[scale=1.2]{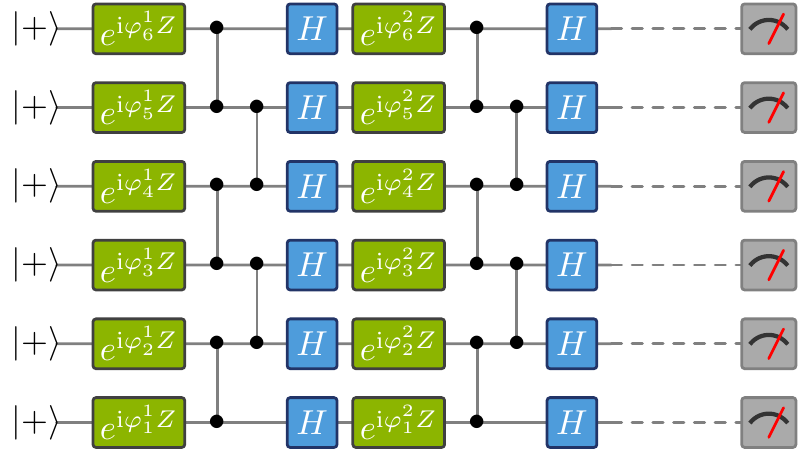}
	\caption{Two layers of the circuit described in Definition~\ref{def:architecture}.}
\end{figure}
First, notice that $e^{\ii H}$ implements a $CZ$-gate (up to a global phase) along each edge.
Furthermore, the preparation state defined in~\eqref{eq:preparation} can be written as
\begin{equation}
|\psi_{\beta}\rangle=\bigotimes_{i=1}^N\left(|0\rangle+e^{\ii\beta_i}|1\rangle\right)=\bigotimes_{i=1}^N R(\beta_i)|+\rangle,
\end{equation}
with $R(\theta) \coloneqq \mathrm{diag}(1,e^{\ii \theta})$.
What is more, the gates $R(\beta_i)$ commute with the Hamiltonian. 
The form of the effective circuit follows from teleporting the gates acting on the columns along the $CZ$ gates between the columns.
Notice that 
\begin{equation}
R(\beta_i)=e^{\ii \frac{\beta_i}{2}}e^{-\ii \frac{\beta_i}{2}Z}.
\end{equation}

\section{$2$-designs from architectures of quantum simulation}
\label{appendix:proof-2-design}

In this appendix, we prove that the architectures for quantum simulation described in Ref.~\cite{bermejo-vega_architectures_2018} form relative $\varepsilon$-approximate $2$-designs as stated in Theorem~\ref{theorem:t-design}. 
In turn, the $2$-design property of the last column implies 
anticoncentration of the state on all qubits as explained in the main text. 
Let us restate the theorem here for convenience. 

\design*

To prove the theorem, we will apply a general proof strategy that has been pioneered in Refs.~\cite{harrow_random_2009,brandao_local_2016}.
The idea of the proof is to successively reduce the $2$-design property to spectral gaps of certain frustration-free Hamiltonians: 

In the first step (App.~\ref{appendix:tpe}), we reduce the $2$-design property to the so-called $2$-copy tensor-product expander property of the probability distribution $v$ with respect to which the random circuit is drawn. 
Technically, this property is an upper bound of the quantity $g(v,2)$ for the probability distribution $v$.
Using the translation-invariance of the measure for our random circuits, one can reduce bounding the quantity $g(v,2)$ for the full measure to bounding the same quantity for the measure for a constant number of steps in the circuit.
In the second step (App.~\ref{app:spectral gap reduction}), we identify a suitable number of circuit steps for which we characterize the corresponding measure $v$.  
We then show that an upper bound for $g(v,2)$ can be reduced to a lower bound for certain frustration-free Hamiltonians.
In the last step (App.~\ref{app:lower bound gap}), we apply methods from quantum many-body theory to obtain a lower bound for the spectral gap of those frustration-free Hamiltonian. 
In particular, we check that those frustration-free Hamiltonians satisfy the conditions required for a famous theorem of \citet{nachtergaele_gap_1994} to hold.

\subsection{Tensor product expanders and designs}\label{appendix:tpe}

In the first step of the proof, we reduce the $2$-design property to a so-called $2$-copy tensor product expander property as established in Ref.~\cite{brandao_local_2016}. 
For completeness and to set the notation, we review this step here. 
We use the following \emph{relative-error} definition of unitary $t$-designs. 
This notion is distinct from an \emph{additive-error} definition of $t$-designs, which would not suffice to prove anticoncentration. 
\begin{definition}[Unitary $t$-designs]
	Let $v$ be a distribution on $\mathbb{U}(N)$. Then, $v$ is an $\varepsilon$-approximate $t$-design if
	\begin{equation}
	(1-\varepsilon)\Delta_{\mu_{\rm Haar},t}\preccurlyeq \Delta_{v,t}\preccurlyeq (1+\varepsilon)\Delta_{\mu_{\rm Haar},t},
	\end{equation}
	where
	\begin{equation}
	\Delta_{v,t}(\rho) \coloneqq \int_{\mathbb{U}(N)}U^{\otimes t}\rho \left(U^{\dagger}\right)^{\otimes t}\mathrm{d}v(U),
	\end{equation}
	and $A\preccurlyeq B$ if and only if $B-A$ is completely positive.
\end{definition}

The basis of our proof is that approximate $t$-designs are closely related to the notion of a quantum $t$-copy tensor product expander (TPE). 
\begin{definition}[\cite{brandao_local_2016}]
	Let $v$ be a distribution on $\mathbb{U}(N)$. $v$ is a $(N,\lambda,t)$ TPE if 
	\begin{equation}
	g(v,t) \coloneqq \left|\left|\int_{\mathbb{U}(N)}U^{\otimes t,t}\mathrm{d}v(U)-\int_{\mathbb{U}(N)}U^{\otimes t,t}\mathrm{d}\mu_{\rm Haar}(U)\right|\right|_{\infty}\leq\lambda,
	\end{equation}
	where we denote $U^{\otimes t,t} \coloneqq U^{\otimes t}\otimes (U^{*})^{\otimes t}$.
\end{definition}

This connection is formalized in the following lemma.
\begin{lemma}[\cite{brandao_local_2016}]\label{lemma:expandertodesign}
	If $g(v,t)\leq\varepsilon$, then $v$ is an $\varepsilon N^{2t}$-approximate $t$-design.
\end{lemma}

A key step of the proof is to use a property of TPEs about how they behave under the concatenation of randomly drawn unitaries.
This allows one to reduce the TPE property for the full circuit to a TPE property of few steps of the circuit.
First, note that for two unitaries $V,U\in \mathbb{U}(N)$ the map $?^{\otimes t,t}$ is an action of the unitary group:
\begin{equation}\label{eq:multiplicativeotimestt}
(UV)^{\otimes t,t}=U^{\otimes t,t}V^{\otimes t,t}.
\end{equation}
If we draw $k$ times independently unitaries from the distribution $v$ and concatenate them, we obtain a unitary drawn from the \textit{$k$-fold convolution} $v^{*k}$. 
The latter is defined as the push-forward measure under the multiplication:
\begin{equation}
v^{*k}(X) \coloneqq \int\chi_{X}(U_1...U_k)\mathrm{d}v(U_1)...\mathrm{d}v(U_n),
\end{equation}
with $\chi_X$ denoting the characteristic function of the set $X\subset \mathbb{U}(N)$.
Notice that this is precisely the situation we are faced with for random local quantum circuits as effected by the MBQC representation of the architectures. By definition, in this case, the measure of the full circuit is a concatenation of local measures for a small number of steps in the circuit. 

Since in the end we want to upper bound $g(v,t)$, the following well-known lemma (see,
e.g., Ref.~\cite{brown_convergence_2010}) capturing this composition property will be at the heart of our proof.
\begin{lemma}[Composition property]
\label{lem:composition property}
	$g(v^{*k},t)\leq g(v,t)^k$.
\end{lemma}
\begin{proof}
	We use the translation-invariance of the Haar measure. Let us denote
	\begin{equation}
	P_v \coloneqq \int_{\mathbb{U}(N)}U^{\otimes t,t}\mathrm{d}v(U), \qquad P_{\rm Haar} \coloneqq \int_{\mathbb{U}(N)}U^{\otimes t,t}\mathrm{d}\mu_{\rm Haar}(U).
	\end{equation}
	Using \eqref{eq:multiplicativeotimestt}, we can now easily see that for any distribution $v$ we have
	\begin{align}
	P_vP_{\rm Haar}&=\left(\int_{\mathbb{U}(N)}U^{\otimes t,t}\mathrm{d}v(U)\right)\left(\int_{\mathbb{U}(N)}U^{\otimes t,t}\mathrm{d}\mu_{\rm Haar}(U)\right)\\
	&=\int_{\mathbb{U}(N)}\int_{\mathbb{U}(N)}(U_1U_2)^{\otimes t,t}\mathrm{d}\mu_{\rm Haar}(U_1)\mathrm{d}v(U_2)\\
	&=\int_{\mathbb{U}(N)}\int_{\mathbb{U}(N)}(U_1)^{\otimes t,t}\mathrm{d}\mu_{\rm Haar}(U_1)\mathrm{d}v(U_2)\\
	&=\int_{\mathbb{U}(N)}(U_1)^{\otimes t,t}\mathrm{d}\mu_{\rm Haar}(U_1)=P_{\rm Haar}.
	\end{align}
	Analogously, we obtain $P_{\rm Haar}P_v=P_{\rm Haar}$. Notice that this also shows $P_{\rm Haar}^2=P_{\rm Haar}$, i.e. $P_{\rm Haar}$ is a projector. Notice that $P_{\rm Haar}$ is in fact an orthogonal projector.
	As a consequence we have
	\begin{align}
	g(v,t)^k&=\left|\left|P_v-P_{\rm Haar}\right|\right|_{\infty}^k\geq\left|\left|(P_v-P_{\rm Haar})^k\right|\right|_{\infty}=\left|\left|(P_v-P_{\rm Haar})^k\right|\right|_{\infty}=\left|\left|P_v^k+\sum_{j=1}^{k}{k\choose j}(-1)^jP_{\rm Haar}\right|\right|_{\infty}\\
	&=\left|\left|P_v^k-P_{\rm Haar}\right|\right|_{\infty}=g(v^{*k},t).
	\end{align}
	We notice that for distributions $v$ such that $\int_{\mathbb{U}(N)}U^{\otimes t,t} \mathrm{d}v(U)$ is normal, even equality holds. 
\end{proof}

\subsection{Reduction to spectral gaps of frustration-free Hamiltonians}
\label{app:spectral gap reduction}

The random quantum circuits generated by measuring the first $m-1$ columns are translation invariant in the sense that the full measure is the $(m-1)$-fold convolution of the measure for an individual column. 
We refer to the circuit generated by one column as a \emph{layer of the random circuit}. 
Clearly, when using Lemma~\ref{lem:composition property} one has the freedom of choosing any $k$-fold convolution of the measure for a single layer since the full circuit is then the $(m-1)/k$-fold convolution of the resulting measure. 

The simplest choice of local measure when proving Theorem~\ref{theorem:t-design} is of course to choose $v$ to be the distribution that is defined by one layer of~\eqref{eq:circuitpicture}.
However, the fixed unitary $E$ in this layer complicates the reduction to spectral gaps. 
On the other hand, the gates of the form $\prod_ie^{\ii\varphi_i Z_i}$ are not entangling and thus can not generate all unitaries, a condition required for the proof. 

A more sophisticated and ultimately successful approach is to group the circuits into subroutines of three layers:
\begin{figure}[H]
	\center
	\includegraphics[scale=1.15]{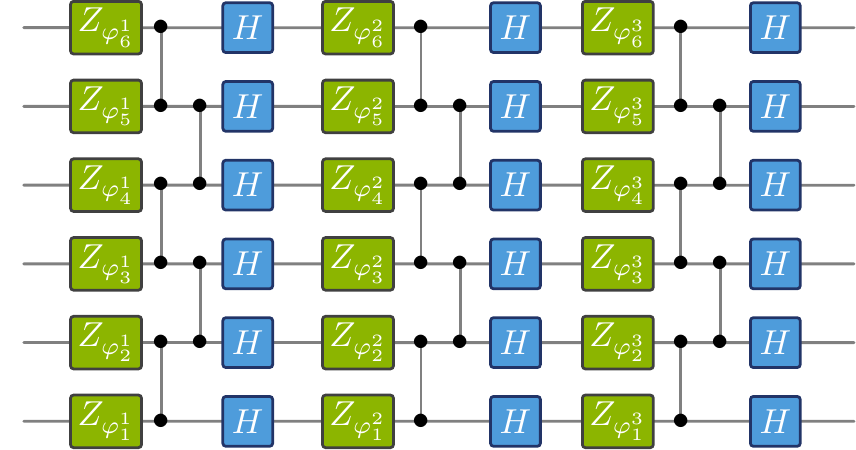}
	\caption{Three layers of the physical circuit on six sites. 
	For a convenient graphical representation, we use the notation $Z_{\theta} \coloneqq \exp(\mathrm{i}\theta Z)$.}
\end{figure}
The global entangling gates $E$ are still present. 
We observe that any such group of three layers can be reformulated as follows.
\begin{figure}[H]
	\center
	\includegraphics[scale=1.15]{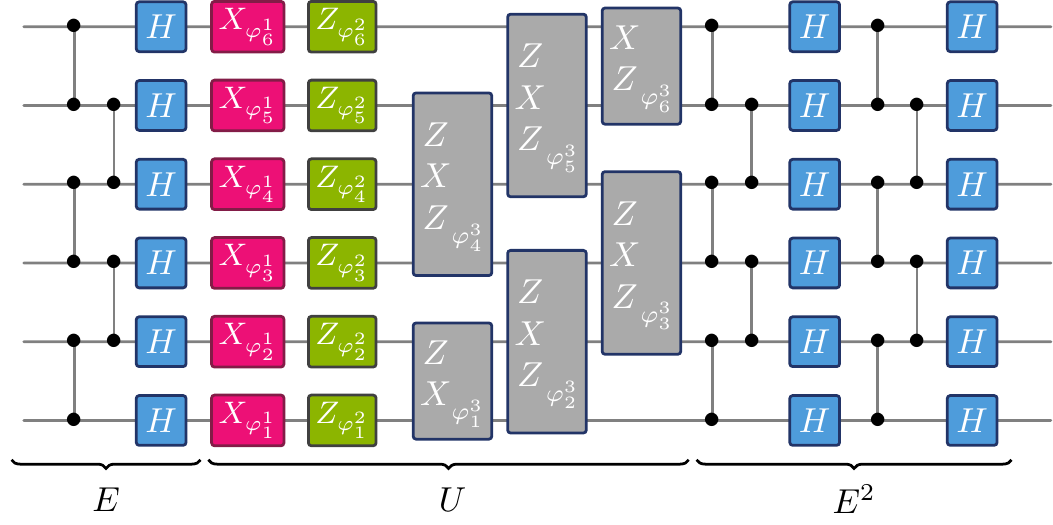}
	\caption{A circuit on six sites that implements the same unitary as the one in Figure 2.}
\end{figure}
Every layer is now of the form $E^2UE$, where $U$ is drawn randomly. 
The following lemma shows that we can simply remove $E$ and $E^2$:
\begin{lemma}[Removal lemma]\label{lemma:fixedunitaries}
	Given a distribution $v$ on $\mathbb{U}(N)$ and two fixed unitaries $V,W\in \mathbb{U}(N)$. 
	Consider the distribution $VvW$ that is defined by drawing $U$ from $v$ and then form $VUW$. Then, 
    $g(VvW,t)= g(v,t)$.
\end{lemma}
\begin{proof}
	The proof follows from invariance of the operator norm under unitaries:
	\begin{align}
	g(v_{V,W},t)&=\left|\left|\int_{\mathbb{U}(N)}(VUW)^{\otimes t,t}\mathrm{d}v(U)-\int_{\mathbb{U}(N)}U^{\otimes t,t}\mathrm{d}\mu_{\rm Haar}(U)\right|\right|_{\infty}\\
	&=\left|\left|\int_{\mathbb{U}(N)}(VUW)^{\otimes t,t}\mathrm{d}v(U)-\int_{\mathbb{U}(N)}(VUW)^{\otimes t,t}\mathrm{d}\mu_{\rm Haar}(U)\right|\right|_{\infty}\\
    &=\left|\left|V^{\otimes t,t}\left(\int_{\mathbb{U}(N)}U^{\otimes t,t}\mathrm{d}v(U)\right)W^{\otimes t,t}-V^{\otimes t,t}\left(\int_{\mathbb{U}(N)}U^{\otimes t,t}\mathrm{d}\mu_{\rm Haar}(U)\right)W^{\otimes t,t}\right|\right|_{\infty}\\
    &=\left|\left|V^{\otimes t,t}\left(\int_{\mathbb{U}(N)}U^{\otimes t,t}\mathrm{d}v(U)-\int_{\mathbb{U}(N)}U^{\otimes t,t}\mathrm{d}\mu_{\rm Haar}(U)\right)W^{\otimes t,t}\right|\right|_{\infty}\\
    &\stackrel{\dagger}{=} \left|\left|\int_{\mathbb{U}(N)}U^{\otimes t,t}\mathrm{d}v(U)-\int_{\mathbb{U}(N)}U^{\otimes t,t}\mathrm{d}\mu_{\rm Haar}(U)\right|\right|_{\infty}=g(v,t).
	\end{align}
	In $\dagger$ we use that $V^{\otimes t,t}$ and $W^{\otimes t,t}$ are unitary operators and that multiplying with unitaries does not change the singular values. However, the operator norm coincides with the norm of the largest singular value.
\end{proof}
Lemma~\ref{lemma:fixedunitaries} allows us to restrict to the distribution $v_{n}$ defined as follows: First, draw randomly $\varphi^X_{i},\varphi^Z_{i},\varphi^{ZXZ}_i\in S^1$ from the Haar measure on $S^1$ and then form the unitary
\begin{equation}
U=\left(e^{\mathrm{i}\varphi^{ZXZ_1} ZXZ_1}... e^{\mathrm{i}\varphi^{ZXZ}_n ZXZ_n}\right) \left(e^{\mathrm{i}\varphi^Z_1 Z_1}... e^{\mathrm{i}\varphi^Z_n Z_n}\right)\left(e^{\mathrm{i}\varphi^X_1 X_1}... e^{\mathrm{i}\varphi^X_n X_n}\right).
\end{equation}
As we will see in Lemma~\ref{lemma:universalgateset} this distribution is universal in the sense of the following subsection:
\begin{definition}[Universal propability distributions]
We call a probability distribution $v$ on the unitary group $\mathbb{U}(N)$ \textit{universal} if for every ball $B_{\varepsilon}\subset\mathbb{U}(N)$, there exists a $k\in\mathbb{N}$ such that $v^{*k}(B_{\varepsilon})>0$.
\end{definition}
For convenience, we define $ZXZ_1 \coloneqq X_1Z_2$ and $ZXZ_n=Z_{n-1}X_{n}$ in the following.
\begin{lemma}\label{lemma:universalgateset}
The gate set consisting of the unitaries $\left\{e^{\ii\varphi^Z_j Z_j}, e^{\ii\varphi^X_j X_j}, e^{\ii\varphi^{ZXZ}_jZXZ_j}\right\}_{j}$ is universal.
\end{lemma}
\begin{proof}
	We first observe that the generated set of $\left\{e^{\ii\varphi^Z_j Z_j}, e^{\ii\varphi^X_j X_j}\right\}$ is dense in the unitaries acting on the system at site $j$. 
	We first consider the boundary: On sites $1$ and $2$, we have gates of the form $e^{\ii\varphi_1^{ZXZ} X_1\otimes Z_2}$. 
	For most $\varphi_1^{ZXZ}$ this is an entangling gate in the sense of Refs.~\cite{brylinski_universal_2001,bremner_universal_2002}.
	Together with the $1$-qubit unitaries, these hence yield a universal gate set on sites $1$ and $2$~\cite{brylinski_universal_2001,bremner_universal_2002}. 
	
	We can now use this universality on the boundary to `propagate universality to the bulk`: 
	Consider the gate $e^{\ii\varphi_2^{ZXZ}Z_1\otimes X_2\otimes Z_3}$. 
	Since we have a universal gate set on site $1$ and $2$, we can use it to approximate the unitary $\mathrm{CNOT}$ which satisfies
	\begin{equation}
	 \mathbb{1}\otimes Z_2=\mathrm{CNOT}(Z_1\otimes Z_2)\mathrm{CNOT}.
	\end{equation}
	Hence, we can build the unitary
	\begin{equation}
	(\mathrm{CNOT}\otimes\mathbb{1})(\mathbb{1}\otimes H\otimes\mathbb{1})e^{\ii\varphi_2^{ZXZ}Z_1\otimes X_2\otimes Z_3}(\mathbb{1}\otimes H\otimes \mathbb{1})(\mathrm{CNOT}\otimes\mathbb{1})=e^{\ii\varphi_2^{ZXZ}\mathbb{1}\otimes Z_2\otimes Z_3},
	\end{equation}
	which is entangling on sites $2$ and $3$. We can now repeat this procedure inductively and thereby propagate universality to all sites.
\end{proof}
The above lemma allows us to keep the asymptotic scaling low:
To obtain universality with a more standard translation-invariant gate set, one would have to group the effective circuit into concatenations of $n$ layers. 
Then, one would obtain the gates $e^{\ii\varphi XZ}$, a $2$-local entangling unitary, on every pair of qubits. 
A similar observation was used in Ref.~\cite{mantri_universal_2016} to show universality of quantum computing with the cluster state and $(X,Y)$-plane measurements.
However, in our case this would lead to a linear overhead in the convergence to secon moments (plus an additional overhead from the detectability lemma). 
Consider the projectors 
\begin{equation}\label{eq:def-projectors}
P^X_i \coloneqq \frac{1}{2\pi}\int \left(e^{\mathrm{i}\varphi^X_i X_i}\right)^{\otimes 2,2}\mathrm{d}\varphi^X_i, \qquad P^Z_i \coloneqq \frac{1}{2\pi}\int \left(e^{\mathrm{i}\varphi^Z_i Z_i}\right)^{\otimes 2,2}\mathrm{d}\varphi^X_i, \qquad P^{ZXZ}_i \coloneqq \frac{1}{2\pi}\int \left(e^{\mathrm{i}\varphi^{ZXZ}_i ZXZ_i}\right)^{\otimes 2,2}\mathrm{d}\varphi^{ZXZ}_i
\end{equation}
and define the Hamiltonian
\begin{equation}\label{eq:hamiltonianspectralgap}
H_n \coloneqq \sum_{i=1}^n\left(\mathbb{1}-P_i^{X}\right)+\sum_{i=1}^n\left(\mathbb{1}-P_i^Z\right)+\sum_{i=1}^n\left(\mathbb{1}-P_i^{ZXZ}\right).
\end{equation}
We would like to reduce the proof to lower-bounding the spectral gap of $H_n$.
 We achieve this by making use of the \textit{detectability lemma}~\cite{aharonov_detectability_2008} in its generalized version~\cite{anshu_detectability_2016}:
\begin{lemma}[\cite{anshu_detectability_2016}]\label{lemma:detectability}
	
Let $\{Q_1,...,Q_m\}$ be a set of projectors and $H=\sum_i Q_i$. 
Assume that each of the $Q_i$ commutes with all but at most $g$ others. Given any state $|\psi^{\perp}\rangle$ orthogonal to the ground state, then 
\begin{equation}
\left|\left|\prod_{i=1}^m(\mathbb{1}-Q_i)|\psi^{\perp}\rangle\right|\right|^2\leq \frac{1}{\frac{\Delta(H)}{g^2}+1}.
\end{equation}
\end{lemma}
We can now prove the following key lemma:
\begin{lemma}[Detectability lemma bound]\label{lemma:bound-detectability}
	\begin{equation}
	g(v_{n},2)\leq \frac{1}{\sqrt{\frac{\Delta(H_n)}{9}+1}},
	\end{equation}
	where $\Delta(H_n)$ denotes the spectral gap of $H_n$, i.e. the difference between its lowest and second-lowest eigenvalue.
\end{lemma}
\begin{proof}
 Consider the operator
	\begin{equation}
	T_n \coloneqq \int_{\mathbb{U}(N)}U^{\otimes 2,2}\mathrm{d}v_n(U)\stackrel{\dagger}{=} \left(P^{ZXZ}_1... P^{ZXZ}_n\right)\left(P^Z_1... P^Z_n\right)\left(P_1^{X}... P_n^X\right),
	\end{equation}
	where $\dagger$ follows from evaluating the integrals over the angles $\varphi_i^X, \varphi_i^Z$ and $\varphi_i^{ZXZ}$ independently.
	Notice that the operator $T_n$ is not necessarily Hermitian or even normal. 
	Thus, there is no reason to assume that $T_n$ has an orthonormal basis of eigenvectors.
	
	Consider the eigenspace $\mathrm{ES}(v)$ of the operator $\int U^{\otimes 2,2}\mathrm{d}v(U)$ to the eigenvalues $1$.
	It was proven in Ref.~\cite[Lemma~3.7]{harrow_random_2009} that for a universal distribution $v$ we have $\mathrm{ES}(v)=\mathrm{ES}(\mu_{\mathrm{Haar}})$.
	 In particular, in combination with the universality of $v_n$ proven in Lemma~\ref{lemma:universalgateset} we obtain
	\begin{align}
	g(v_{n},2)^2&=\left|\left|\int_{\mathbb{U}(N)}U^{\otimes 2,2}\mathrm{d}v_n(U)-\int_{\mathbb{U}(N)}U^{\otimes 2,2}\mathrm{d}\mu_{\rm Haar}(U)\right|\right|^2_{\infty}\\
	&=\left|\left|\int_{\mathbb{U}(N)}U^{\otimes 2,2}\mathrm{d}v_n(U)-\mathrm{P}_{\mathrm{ES}(\mu_{\mathrm{Haar}})}\right|\right|^2_{\infty}\\
	&=\left|\left|\int_{\mathbb{U}(N)}U^{\otimes 2,2}\mathrm{d}v_n(U)-\mathrm{P}_{\mathrm{ES}(v)}\right|\right|^2_{\infty}\\
		&= \max_{\substack{\psi\\||\psi||=1}}\left|\left|(T_n-\mathrm{P}_{\mathrm{ES}(v)})\left|\psi\right\rangle\right|\right|^2\\
	&\stackrel{\S}{\leq}\max_{\substack{\psi\in\mathrm{ES}(v)^{\perp}\\||\psi||=1}}\||T_n|\psi\rangle||^2,
	\end{align}	
	where $\mathrm{P}_{\mathrm{ES(v)}}$ denotes the orthogonal projector onto $\mathrm{ES}(v)$.
	In $\S$ we decompose $|\psi\rangle=|\psi^{||}\rangle+|\psi^{\perp}\rangle$ with $|\psi^{||}\rangle\in \mathrm{ES}(v)$ and $|\psi^{\perp}\rangle\in \mathrm{ES}(v)^{\perp}$
	and compute
	\begin{align}
	 g(v_n,2)^2&=\max_{\psi\neq 0}\frac{\left|\left|(T_n-\mathrm{P}_{\mathrm{ES}(v)})\left(|\psi^{||}\rangle+|\psi^{\perp}\rangle\right)\right|\right|^2}{||\psi^{||}+\psi^{\perp}||^2}\\
	 &=\max_{\substack{\psi\\\psi^{\perp}\neq 0}}\frac{\left|\left|(T_n-\mathrm{P}_{\mathrm{ES}(v)})|\psi^{\perp}\rangle\right|\right|^2}{||\psi^{||}||^2+||\psi^{\perp}||^2}\\
	 &\leq \max_{\psi^{\perp}\neq 0}\frac{\left|\left|T_n|\psi^{\perp}\rangle\right|\right|^2}{||\psi^{\perp}||^2}.
	\end{align}	
	Moreover, each of the local projectors $P^X_i,P^Z_i, P^{ZXZ}_i$ commutes with all but at most $3$ others. 
	The detectability lemma (Lemma~\ref{lemma:detectability}), thus, directly yields the bound for $g=3$.
\end{proof}

\subsection{Lower bounding the spectral gap}
\label{app:lower bound gap}

In the following we are going to show that there is a constant $\alpha>0$ such that $\Delta(H_n)>\alpha$ for all $n$. 
One method to obtain such a lower bound on the spectral gap in the thermodynamic limit is the \textit{Nachtergaele bound}~\cite{nachtergaele_gap_1994} sometimes called \textit{martingale method}.
 
\begin{lemma}[\citet{nachtergaele_gap_1994}]\label{lemma:Nachtergaele}
Given a family of Hamiltonians $H_{[p,q]}$ for $[p,q]\subset \mathbb{Z}$ acting on $(\mathbb{C}^2)^{\otimes |[p,q]|}=(\mathbb{C}^2)^{\otimes (q-p+1)}$. Assume there are numbers positive numbers $l$, $d_l$, $q_l$ and $\gamma_l$ such that the following conditions hold:
\begin{enumerate}
 \item There is a constant $d_l$ for which the Hamiltonians satisfy
 \begin{equation}\label{eq:firstcondition}
 0\leq \sum_{i=l}^N \id_{[1,i-l]}\otimes H_{[i-l+1,i]}\otimes\id_{[i+1,n]}\leq d_l H_{[1,n]}\qquad \text{for all}\qquad n\geq q_l.
 \end{equation}
 \item The lowest eigenvalue for all $H_{[p,q]}$ is $0$ and there is a spectral gap $\gamma_l>0$:
 \begin{equation}
 \Delta\left(H_{[q-l+1,q]}\right)\geq \gamma_l \,\;\;\;\text{for all}\;\;\;\;\;q\geq q_l
 \end{equation}
 for some constant $q_l$.
 \item We denote the ground state projector of $\id_{[1,p-1]}\otimes H_{[p,q]}\otimes \id_{[q+1,n]}$ with $G_{[p,q]}$. There exist $\varepsilon_l<1/\sqrt{l+1}$ such that
 \begin{equation}\label{eq:condition3}
 \left|\left|G_{[q-l+1,q+1]}\left(G_{[1,q]}-G_{[1,q+1]}\right)\right|\right|_{\infty}\leq\varepsilon_l\qquad\text{for all}\qquad q\geq q_l.
 \end{equation}
\end{enumerate}
Then,
\begin{equation}
\Delta\left(H_{[1,n]}\right)\geq \frac{\gamma_{l+1}}{d_{l+1}}\left(1-\varepsilon_l\sqrt{l+1}\right)^2\qquad \text{for all}\qquad n\geq q_l.
\end{equation}
\end{lemma}

Here, we would like to apply Lemma~\ref{lemma:Nachtergaele} to the family defined in Eq.~\eqref{eq:hamiltonianspectralgap}.
Notice that the conditions in Lemma~\ref{lemma:Nachtergaele} do not require translation-invariance.
The most non-trivial part will be to verify the last condition, which requires information about the ground spaces and the ground space projectors. 
We need the ground spaces of the following Hamiltonians:
\begin{equation}
H^L_n \coloneqq H_n-\left(\mathbb{1}-P^{ZXZ_n}\right)\qquad\text{and}\qquad H^B_n \coloneqq H_n-\left(\mathbb{1}-P^{ZXZ_1}\right)-\left(\mathbb{1}-P^{ZXZ_n}\right)
\end{equation}
with $L$ standing for `left edge` and $B$ standing for `bulk`. 
In the following lemma we identify 
$$
\left(\left(\mathbb{C}^2\right)^{\otimes 2,2}\right)^{\otimes n}\stackrel{\sim}{=}\left(\left(\mathbb{C}^2\right)^{\otimes n}\right)^{\otimes 2,2}
$$
to ease the notation.
\begin{lemma}[Ground spaces]\label{lemma:groundspace}
	We denote with $\mathcal{G}_n, \mathcal{G}^L_n$ and $\mathcal{G}_n^B$ the ground space of the Hamiltonians $H_n, H^L_n$ and $H^B_n$. 
	We have
	\begin{equation}
	\mathcal{G}_n=\mathcal{G}^L_n=\left\{\ket{\psi_0}^{\otimes n},\ket{\psi_1}^{\otimes n}\right\}\qquad\text{and}\qquad \mathcal{G}^B_n=\left\{\ket{\psi_0}^{\otimes n},\ket{\psi_1}^{\otimes n},\ket{\psi_{\phi}}^{\otimes n}\right\},
	\end{equation}
    with
	\begin{equation}\label{eq:groundstates1}
	\ket{\psi_0} \coloneqq (\id\otimes V_{\pi_0})|\Phi\rangle=\frac12(\ket{0000}+\ket{0101}+\ket{1010}+\ket{1111}),
	\end{equation}
	\begin{equation}
   \ket{\psi_1} \coloneqq (\id\otimes V_{\pi_1})|\Phi\rangle=\frac12(\ket{0000}+\ket{0110}+\ket{1001}+\ket{1111}),
   \end{equation}
   where $|\Phi\rangle=\frac{1}{2}\sum_{i\in\{0,1\}^2}|i\,
   i\rangle$ is the maximally entangled state vector and $V_{\pi_i}$ are the standard representations of the identity and the swap $\pi_0,\pi_1\in S_2$.
	The third basis state vector is
	\begin{equation}\label{eq:groundstates2}
	\ket{\psi_{\phi}} \coloneqq \ket{\psi_0}-\ket{\psi_1}=\frac12(|0101\rangle+|1010\rangle-|0110\rangle-|1001\rangle).
	\end{equation}
\end{lemma}
\begin{proof}
	We have
	\begin{align}
	&H_n\ket\varphi=0\iff P^{X_i}\ket\varphi=\ket\varphi,\quad P^{Z_i}\ket\varphi=\ket\varphi,\quad P^{ZXZ_i}\ket\varphi=\ket\varphi\quad \forall j\\
	 &\iff\left(e^{\mathrm{i}\varphi^X_i X_i}\right)^{\otimes 2,2}\ket\varphi=\ket\varphi \quad\wedge\quad \left(e^{\mathrm{i}\varphi^Z_i Z_i}\right)^{\otimes 2,2}\ket\varphi=\ket\varphi\quad\wedge\quad  \left(e^{\mathrm{i}\varphi^{ZXZ}_i ZXZ_i}\right)^{\otimes 2,2}\ket\varphi=\ket\varphi\quad\forall\; \varphi^X_i,\varphi^Z_i,\varphi^{ZXZ}_i
	\end{align}
	and similarly for $H^L_n$ and $H^B_n$.
	The first equivalence follows from the fact that all local Hamiltonians are positive semi-definite and the second follows similar to the proof of Ref.~\cite[Lemma\,17]{brandao_local_2016} (originally~\cite{brown_convergence_2010}):
	\begin{equation}
	\mathrm{Re}\left\langle\phi\left|P^{Z_i}\right|\phi\right\rangle=\mathrm{Re}\mathbb{E}_{\varphi}\left\langle\phi\left|\left(e^{\ii\varphi Z_i}\right)^{\otimes 2,2}\right|\phi\right\rangle\leq 1,
	\end{equation}
	with equality if and only if $\left(e^{\ii\varphi Z_i}\right)^{\otimes 2,2}|\phi\rangle=|\phi\rangle$ for all but a measure zero subset, which is empty by continuity of the angles.
	This follows analogously for gates of the form $e^{\ii\varphi X_i}$ and $e^{\ii\varphi ZXZ_i}$.
	 The form of $\mathcal{G}_n$ and $\mathcal{G}^L_{n}$ follows from the fact the gate sets defining the Hamiltonians $H_n$ and $H_n^L$ generate a dense subset of the unitaries. 
	 Thus, ground states $|\varphi\rangle$ satisfy $U^{\otimes 2,2}|\varphi\rangle=|\varphi\rangle$ for all unitaries $U\in \mathbb{U}(2^n)$.
	 This problem is well studied in the context of the Schur-Weyl duality. 
	 In fact, it is well known (compare, e.g., 
	 the proof of Ref.~\cite[Lemma~17]{brandao_local_2016}) that this implies the form presented in Lemma~\ref{lemma:groundspace}.
	
	The more involved case is $\mathcal{G}^B_n$. 
	The ground space of the local Hamiltonians $(\mathbb{1}-P^Z)+(\mathbb{1}-P^X)$ is generated by $\ket{\psi_0}$ and $\ket{\psi_1}$. 
	Thus, the ground space of
	$
	\sum_{i=1}^n(\mathbb{1}-P_i^Z)+(\mathbb{1}-P_i^X)
	$
	is given by $\mathrm{span}\left\{\prod_{i=1}^n\ket{\psi_{s_i}}|s\in\{0,1\}^n\right\}$.
	Furthermore, notice that 
	\begin{align}
	&\left(e^{\mathrm{i}\varphi^X_i X_i}\right)^{\otimes 2,2}\ket\varphi=\ket\varphi \quad\wedge\quad \left(e^{\mathrm{i}\varphi^Z_i Z_i}\right)^{\otimes 2,2}\ket\varphi=\ket\varphi\quad\wedge\quad  \left(e^{\mathrm{i}\varphi^{ZXZ}_i ZXZ_i}\right)^{\otimes 2,2}\ket\varphi=\ket\varphi\quad\forall\; \varphi^X_i,\varphi^Z_i,\varphi^{ZXZ}_i\\
	&\iff	\left(e^{\mathrm{i}\varphi^X_i X_i}\right)^{\otimes 2,2}\ket\varphi=\ket\varphi \quad\wedge\quad \left(e^{\mathrm{i}\varphi^Z_i Z_i}\right)^{\otimes 2,2}\ket\varphi=\ket\varphi\quad\wedge\quad  \left(e^{\mathrm{i}\varphi^{ZZZ}_i ZZZ_i}\right)^{\otimes 2,2}\ket\varphi=\ket\varphi\quad\forall\; \varphi^X_i,\varphi^Z_i,\varphi^{ZZZ}_i.
	\end{align}
   We use that the gates $e^{\ii\varphi^X X}$ and $e^{\ii\varphi^Z Z}$ generate the single qubit unitaries. 
   Thus, up to infinitesimal error, they generate the Hadamard gate $H$, which we can use to rotate the tensor factor $X$ to $Z$.
   Hence, we have the additional constraint that 
	\begin{equation}\label{eq:lastcondition}
	\left(e^{\mathrm{i}\varphi^{ZZZ}_i ZZZ_i}\right)^{\otimes 2,2}\ket{\varphi}=\ket\varphi
	\end{equation}
	for all $2\leq i\leq n-1$.
	
	First, for three qubits we show in the following that the space is generated by $\ket{\psi_{0}}\ket{\psi_{0}}\ket{\psi_{0}}$, $\ket{\psi_{1}}\ket{\psi_{1}}\ket{\psi_{1}}$ and $\ket{\psi_{\phi}}\ket{\psi_{\phi}}\ket{\psi_{\phi}}$. 
	In more detail, we compute
	\begin{align}
	&\left(e^{\mathrm{i}\varphi ZZZ}\right)^{\otimes 2,2}|0000\rangle|0101\rangle|0110\rangle=
	e^{\ii(+\varphi+\varphi+\varphi+\varphi)}|0000\rangle|0101\rangle|0110\rangle
	=e^{-4\varphi\ii}|0000\rangle|0101\rangle|0110\rangle,\\
	   &\left(e^{\mathrm{i}\varphi ZZZ}\right)^{\otimes 2,2}|0000\rangle|0101\rangle|0101\rangle=e^{\ii(+\varphi+\varphi-\varphi-\varphi)}|0000\rangle|0101\rangle|0101\rangle=|0000\rangle|0101\rangle|0101\rangle,\\
	   &\left(e^{\mathrm{i}\varphi ZZZ}\right)^{\otimes 2,2}|0000\rangle|0000\rangle|0110\rangle=e^{\ii(+\varphi-\varphi+\varphi-\varphi)}|0000\rangle|0000\rangle|0110\rangle=|0000\rangle|0000\rangle|0110\rangle,\\
	  	  &\left(e^{\mathrm{i}\varphi ZZZ}\right)^{\otimes 2,2}|0101\rangle|0101\rangle|0101\rangle=e^{\ii(+\varphi-\varphi-\varphi+\varphi)}|0101\rangle|0101\rangle|0101\rangle=|0101\rangle|0101\rangle|0110\rangle,\\
	  	     &\left(e^{\mathrm{i}\varphi ZZZ}\right)^{\otimes 2,2}|0110\rangle|0110\rangle|0110\rangle=e^{\ii(+\varphi-\varphi+\varphi-\varphi)}|0110\rangle|0110\rangle|0110\rangle=|0110\rangle|0110\rangle|0110\rangle,\\
	  	  	  	  &\left(e^{\mathrm{i}\varphi ZZZ}\right)^{\otimes 2,2}|0101\rangle|0110\rangle|0110\rangle=e^{\ii(+\varphi-\varphi+\varphi+\varphi)}|0101\rangle|0110\rangle|0110\rangle=|0101\rangle|0110\rangle|0110\rangle,\\
	  	  	  	   &\left(e^{\mathrm{i}\varphi ZZZ}\right)^{\otimes 2,2}|0101\rangle|0101\rangle|0110\rangle=e^{\ii(+\varphi-\varphi+\varphi-\varphi)}|0101\rangle|0101\rangle|0110\rangle=|0101\rangle|0101\rangle|0110\rangle.
    \end{align}
	All other cases can be obtained from the fact that $\left(e^{\mathrm{i}\varphi ZZZ}\right)^{\otimes 2,2}$ is invariant under permutations of the qubits and that flipping all four bits in one qubit only changes the sign of the exponent.
	
	The lesson to be learned from the above calculation is that in a linear combination in $\mathrm{span}
	\{\prod_{i=1}^3\ket{\psi_{s_i}}|s\in\{0,1\}^3
	\}$ that satisfies~\eqref{eq:lastcondition}, terms as $|0000\rangle|0101\rangle|0110\rangle$ need to cancel.
	No such terms appear in $|\psi_0\rangle|\psi_0\rangle|\psi_0\rangle$ and $|\psi_1\rangle|\psi_1\rangle|\psi_1\rangle$. 
	Furthermore, consider the linear combination:
	\begin{align*}
	\ket{\varphi} \coloneqq &\lambda_{000}|\psi_{0}\rangle|\psi_{0}\rangle|\psi_{0}\rangle+\lambda_{100}|\psi_{1}\rangle|\psi_{0}\rangle|\psi_{0}\rangle+\lambda_{010}|\psi_{0}\rangle|\psi_{1}\rangle|\psi_{0}\rangle+\lambda_{110}|\psi_{1}\rangle|\psi_{1}\rangle|\psi_{0}\rangle\\
	&+\lambda_{001}|\psi_{0}\rangle|\psi_{0}\rangle|\psi_{1}\rangle+\lambda_{101}|\psi_{1}\rangle|\psi_{0}\rangle|\psi_{1}\rangle+\lambda_{011}|\psi_{0}\rangle|\psi_{1}\rangle|\psi_{1}\rangle+\lambda_{111}|\psi_{1}\rangle|\psi_{1}\rangle|\psi_{1}\rangle.
	\end{align*}
    From the condition that certain terms need to cancel, we obtain $\lambda_{100}=-\lambda_{101}, \lambda_{101}=-\lambda_{001}, \lambda_{001}=-\lambda_{011}, \lambda_{011}=-\lambda_{010}, \lambda_{010}=-\lambda_{110}$. 
    Assuming that $\lambda_{001}\neq 0$, we can w.l.o.g. choose the free parameters $\lambda_{001}=-1$, $\lambda_{000}=1$ and $\lambda_{111}=-1$. This fixes the state
	\begin{align*}
\ket{\varphi}=&|\psi_{0}\rangle|\psi_{0}\rangle|\psi_{0}\rangle-|\psi_{1}\rangle|\psi_{0}\rangle|\psi_{0}\rangle-|\psi_{0}\rangle|\psi_{1}\rangle|\psi_{0}\rangle+|\psi_{1}\rangle|\psi_{1}\rangle|\psi_{0}\rangle+|\psi_{0}\rangle|\psi_{0}\rangle|\psi_{1}\rangle-|\psi_{1}\rangle|\psi_{0}\rangle|\psi_{1}\rangle-|\psi_{0}\rangle|\psi_{1}\rangle|\psi_{1}\rangle-|\psi_{1}\rangle|\psi_{1}\rangle|\psi_{1}\rangle\\
=&|\psi_{\phi}\rangle|\psi_{\phi}\rangle|\psi_{\phi}\rangle.
\end{align*}
	
    We extend this via complete induction over the number of qubits:
    Assume we have the ground state space $\mathcal{G}_n^B$ for the Hamiltonian $H_n$. 
    Using the positivity of $H_n\otimes\id$, we have for any ground state $|\varphi\rangle$ of $H_{n+1}$ the form
	\begin{align*}
	|\varphi\rangle&=\sum_{\sigma\in\{0,1,\phi\}}\ket{\psi_{\sigma}}^{\otimes n}\otimes \left(\lambda_{\sigma}|\psi_0\rangle+\mu_{\sigma}\ket{\psi_1}\right)\\
	&=\lambda_0\ket{\psi_{0}}^{\otimes n} |\psi_0\rangle+\mu_{0}\ket{\psi_0}^{\otimes n}\ket{\psi_1}+\lambda_1\ket{\psi_{1}}^{\otimes n} |\psi_0\rangle+\mu_{1}\ket{\psi_1}^{\otimes n}\ket{\psi_1}+\tilde{\lambda}_{\phi}\ket{\psi_{\phi}}^{\otimes n} |\psi_{\phi}\rangle+\tilde{\mu}_{\phi}\ket{\psi_{\phi}}^{\otimes n}(\ket{\psi_0}+\ket{\psi_1}),
	\end{align*}
	for some constants $\tilde{\lambda}_{\phi}$ and $\tilde{\mu}_{\phi}$.
    These need to satisfy the last constraint~\eqref{eq:lastcondition} for $i=n$ as well.
    Applying $\id-P^{ZZZ}_n$ to this yields
    \begin{align}
    &\left(\id-P^{ZZZ}_n\right)\left(\mu_{0}\ket{\psi_0}^{\otimes n}\ket{\psi_1}+\lambda_1\ket{\psi_{1}}^{\otimes n} |\psi_0\rangle+\tilde{\mu}_{\phi}\ket{\psi_{\phi}}^{\otimes n}(\ket{\psi_0}+\ket{\psi_1})\right)=0\\
    &\iff \mu_0|\psi_{0}\rangle^{\otimes n-2}	\left(\id-\left(e^{\mathrm{i}\varphi^{ZZZ}_n ZZZ_n}\right)^{\otimes 2,2}\right)|\psi_0\rangle^{\otimes 2}|\psi_1\rangle+\lambda_1|\psi_{1}\rangle^{\otimes n-2}	\left(\id-\left(e^{\mathrm{i}\varphi^{ZZZ}_n ZZZ_n}\right)^{\otimes 2,2}\right)|\psi_1\rangle^{\otimes 2}|\psi_0\rangle+\\
    &+\tilde{\mu}_{\phi}|\psi_{\phi}\rangle^{\otimes n-2}	\left(\id-\left(e^{\mathrm{i}\varphi^{ZZZ}_n ZZZ_n}\right)^{\otimes 2,2}\right)|\psi_{\phi}\rangle^{\otimes 2}(|\psi_{0}\rangle+|\psi_1\rangle)=0
    \end{align}
    for all $\varphi^{ZZZ}_n$.
    Using the definitions~\eqref{eq:groundstates1} and~\eqref{eq:groundstates2} one can easily verify that this implies $\mu_0,\lambda_1,\tilde{\mu}_{\phi}=0$:
    The above expression contains the summands $\mu_{0}\left(1-e^{4\ii\varphi^{ZZZ}_n}\right)|\psi_0\rangle^{\otimes n-2}|0000\rangle|0101\rangle|0110\rangle$, $\lambda_1\left(1-e^{4\ii\varphi^{ZZZ}_n}\right)|\psi_1\rangle^{\otimes n-2}|0000\rangle|0110\rangle|0101\rangle$ and $\tilde{\mu}_{\phi}\left(1-e^{4\ii\varphi^{ZZZ}_n}\right)|\psi_{\phi}\rangle^{\otimes n-2}|0101\rangle|0110\rangle|0000\rangle$.
    This completes the proof.
\end{proof}
The last step of the proof of Theorem~\ref{theorem:t-design} is achieved by a lower bound on the spectral gap of $H_n$. 
\begin{lemma}[Lower bound on the spectral gap]
$	\Delta(H_n)\geq \frac{\Delta\left(H^B_7\right)}{32}$ for all $n\geq 8$.
\end{lemma}
\begin{proof}
	For the proof we verify all three conditions in Lemma~\ref{lemma:Nachtergaele}.
	
	\textbf{First condition}. 
We start with the first condition in Lemma~\ref{lemma:Nachtergaele}.
There is a technical difficulty in defining the family of Hamiltonians $H_{[q,p]}$: 
If we simply choose $H_{[q,p]}$ to be $H_{p-q+1}$ acting on the qubits in $[q,p]$ we run into problems with the first condition. 
The problem is that in the sum in~\eqref{eq:firstcondition} the boundary terms act on the bulk and cannot be bounded by the bulk Hamiltonian.
Instead, we reduce the proof to the Hamiltonian $H^L_n$ and choose $H_{[q,p]}$ to be the local summands of $H^L_n$ that act on the qubits in $[q,p]$.
We show that it suffices to lower bound the spectral gap of $H^L_n$:
\begin{align}
\Delta(H_n)&=\min_{\substack{\ket\varphi\in\mathcal{G}_n^{\perp}\\ ||\varphi||=1}}\langle\varphi|H_n|\varphi\rangle\\
&=\min_{\substack{\ket\varphi\in\left(\mathcal{G}^L_n\right)^{\perp}\\ ||\varphi||=1}}\left\langle\varphi\left|H^L_n+\left(\mathbb{1}-P^{ZXZ}_n\right)\right|\varphi\right\rangle\\
&=\min_{\substack{\ket\varphi\in\left(\mathcal{G}^L_n\right)^{\perp}\\ ||\varphi||=1}}\left(\left\langle\varphi\left|H^L_n\right|\varphi\right\rangle+\left\langle\varphi\left|\left(\mathbb{1}-P^{ZXZ}_n\right)\right|\varphi\right\rangle\right)\\
&\geq\min_{\substack{\ket\varphi\in\left(\mathcal{G}^L_n\right)^{\perp}\\ ||\varphi||=1}}\left\langle\varphi\left|H^L_n\right|\varphi\right\rangle=\Delta\left(H^L_n\right).
\end{align}
In the remainder of the proof we will verify the conditions of Lemma~\ref{lemma:Nachtergaele} for the family of Hamiltonians defined as 
$$
H_{[p,q]} \coloneqq \begin{cases}& H^L_{q-p+1}\otimes \id_{[q+1,n]}\quad \text{for}\quad p=1\\
&\id_{[1,p-1]}\otimes H^B_{q-p+1}\otimes \id_{[q+1,n]}\quad \text{for}\quad p>1\end{cases}.
$$

With this definition, we immediately have that the first condition is satisfied obviously with $d_l=(l+1)$.

\textbf{Second condition}. The second condition follows from the fact that $H^L_n$ is frustration-free. 
Because of the translation-invariance in the bulk, we always have
\begin{equation}
 \Delta\left(H_{[n-l+1,n]}\right)=\gamma_{l} \coloneqq \Delta\left(H^B_{l}\right)>0\qquad \text{for all}\qquad n\geq l+1.
 \end{equation}

\textbf{Third condition}. Denote with $G_{n}, G^L_{n}$ and $G^B_n$ the ground space projectors onto the ground spaces $\mathcal{G}_n, \mathcal{G}^L_n$ and $\mathcal{G}^B_n$.
In Lemma~\ref{lemma:groundspace} we identified a basis for these ground spaces. 
The strategy is to use this basis to construct approximations $X^{B/L}$ to the ground space projectors $G^{B/L}$.
For these approximations, we can then verify the third condition directly.
We illustrate how the terms appearing in condition~\eqref{eq:condition3} act on a chain of qubits in Fig.~\ref{fig: nachtergaele condition 3}. 
\begin{figure}[t]
	\center
	\includegraphics[scale=0.9]{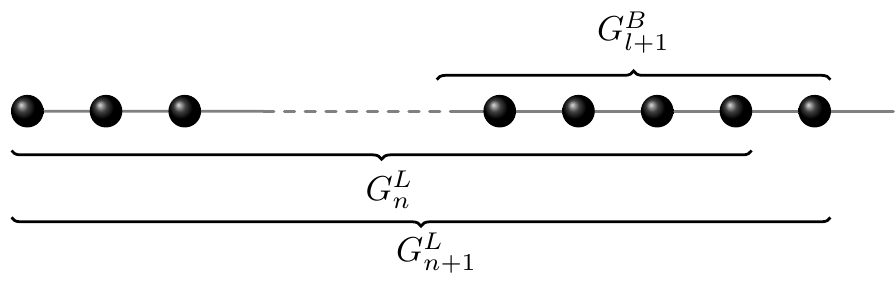}
	\caption{
	\label{fig: nachtergaele condition 3}
	Illustration of the supports of the projectors $G_{[p,q]}$ appearing in Eq.~\eqref{eq:condition3} for a chain of qubits. 
	 }
\end{figure}

 First, observe that 
\begin{equation}
|\langle\psi_{\sigma}|\psi_0\rangle|^k+|\langle\psi_{\sigma}|\psi_1\rangle|^k+|\langle\psi_{\sigma}|\psi_{\phi}\rangle|^k=1+\frac{2}{2^k},
\end{equation}
for all $\sigma\in\{0,1,\phi\}$. In particular, $\{|\psi_0\rangle^{\otimes k},|\psi_1\rangle^{\otimes k},|\psi_{\phi}\rangle^{\otimes k}\}$ constitutes an orthonormal basis up to an exponentially small error.
 We can use this to show that 
\begin{equation}\label{eq:approxprojectors}
\left|\left|\sum_{\sigma\in\{0,1,\phi\}}(|\psi_{\sigma}\rangle\langle\psi_{\sigma}|)^{\otimes k}-G^B_k\right|\right|_{\infty}\leq \frac{\sqrt{6}}{2^{k}}\qquad\text{and}\qquad \left|\left|\sum_{\sigma\in\{0,1\}}(|\psi_{i}\rangle\langle\psi_{i}|)^{\otimes k}-G^L_k\right|\right|_{\infty}\leq \frac{\sqrt{2}}{2^{k}}.
\end{equation}
We only prove the former. The latter can be shown analogously. 
Consider the operator 
\begin{equation}\label{eq:operatorB}
B \coloneqq \sum_{\sigma\in\{0,1,\phi\}}|\sigma\rangle\langle\psi_{\sigma}|^{\otimes k},
\end{equation}
where $|\sigma\rangle$ is any orthonormal basis of $\mathcal{G}^B_k$. Then we have 
\begin{align}
\left|\left|BB^{\dagger}-G^B_{k}\right|\right|_{\infty}&=\left|\left|\sum_{\sigma,\rho\in\{0,1,\phi\}}|\sigma\rangle\langle\psi_{\sigma}|^{\otimes k}|\psi_{\rho}\rangle^{\otimes k}\langle\rho|-\sum_{\sigma\in\{0,1,\phi\}}|\sigma\rangle\langle\sigma|\right|\right|_{\infty}\\
&=\left|\left|\sum_{\sigma\neq\rho}|\langle\psi_{\sigma}|\psi_{\rho}\rangle|^{ k}|\sigma\rangle\langle\rho|\right|\right|_{\infty}\\
&\leq\left|\left|\sum_{\sigma\neq\rho}|\langle\psi_{\sigma}|\psi_{\rho}\rangle|^{ k}|\sigma\rangle\langle\rho|\right|\right|_{2}\\
&=\left({\sum_{\sigma\neq\rho}|\langle\psi_{\sigma}|\psi_{\rho}\rangle|^{2k}}
\right)^{1/2}
=\frac{\sqrt{6}}{2^{k}},
\end{align}
where $\|.\|_{2}$ denotes the Schatten $2$-norm. 
We used for the inequality the monotonicity of the Schatten $p$-norms.
Notice that $B$ and $B^{\dagger}$ act as zero on the orthogonal complement of the ground space. 
Furthermore, both maps are invertible if restricted to the ground space for $n>1$ as $\{|\psi_{\sigma}\rangle\}$ constitutes a basis of the ground space.
Hence,
 $B^{\dagger}B=\sum_{\sigma}(|\psi_{\sigma}\rangle\langle\psi_{\sigma}|)^{\otimes k}$ has the same eigenvalues as $BB^{\dagger}$. 
The above is a bound on the difference between the eigenvalues of $BB^{\dagger}$ restricted to the ground space and $1$, as $BB^{\dagger}$ and $G^B_k$ can be diagonalized simultaneously.
Moreover, $G_k^B$ is the projector onto the support of $B^{\dagger}B$ and we obtain 
\begin{equation}
\left(1-\frac{\sqrt{6}}{2^{k}}\right)G^{B}_k\leq B^{\dagger}B\leq\left(1+\frac{\sqrt{6}}{2^{k}}\right) G^{B}_k.
\end{equation} 
Thus, we have
\begin{equation}
\left|\left|B^{\dagger}B-G^{B}_k\right|\right|_{\infty}\leq\frac{\sqrt{6}}{2^{k}}.
\end{equation}
In the following we denote 
\begin{equation}
X^L_k \coloneqq \sum_{\sigma\in\{0,1\}}(|\psi_{i}\rangle\langle\psi_{i}|)^{\otimes k},\qquad\text{and}\qquad X^B_k \coloneqq \sum_{\sigma\in\{0,1,\phi\}}(|\psi_{\sigma}\rangle\langle\psi_{\sigma}|)^{\otimes k}.
\end{equation}
We can now verify~\eqref{eq:condition3}. 
We successively apply the identity $G^{B/L}=X^{B/L}+\left(G^{B/L}-X^{B/L}\right)$ and the triangle inequality to obtain
\begin{align}
R& \coloneqq \left|\left|G^B_{[n-l+1,n+1]}\left(G^L_{[1,n]}-G^L_{[1,n+1]}\right)\right|\right|_{\infty}\\
&\leq \left|\left|X^B_{[n-l+1,n+1]}\left(X^L_{[1,n]}-X^L_{[1,n+1]}\right)\right|\right|_{\infty}+\frac{11}{2^{l+1}}\\
&=\left|\left|\sum_{i\in\{0,1\}}(|\psi_i\rangle\langle\psi_i|)^{\otimes (n-l)}\otimes Y_i\right|\right|_{\infty}+\frac{11}{2^{l+1}}
\end{align}
with 
\begin{equation}
Y_i \coloneqq \sum_{\substack{\sigma\in\{0,1,\phi\}\\\sigma\neq i}}\left((|\psi_{\sigma}\rangle\langle \psi_{\sigma}|)^{\otimes l}(|\psi_{i}\rangle\langle \psi_{i}|)^{\otimes l}\right)\otimes (|\psi_{\sigma}\rangle\langle\psi_{\sigma}|)\left(\id_{n+1}-|\psi_{i}\rangle\langle\psi_i|\right).
\end{equation}
We get
\begin{align}
R&\leq \left(1+\frac{\sqrt{2}}{2^{n-l}}\right)\max_i||Y_i||_{\infty}+\frac{11}{2^{l+1}}\\
&\leq \left(1+\frac{\sqrt{2}}{2^{n-l}}\right)\sum_{\substack{\sigma\in\{0,1,\sigma\}\\\sigma\neq i}}|\langle\psi_{i}|\psi_{\sigma}\rangle|^{l}+\frac{11}{2^{l+1}}\\
&=\left(1+\frac{\sqrt{2}}{2^{n-l}}\right)\frac{2}{2^{l}}+\frac{11}{2^{l+1}}\leq \frac{8.5}{2^{l}},
\end{align}
for $n\geq l+2$.
We achieve $\varepsilon_l\leq\frac{1}{2\sqrt{l+1}}<\frac{1}{\sqrt{l+1}}$ by setting $l=6$.
The first inequality follows as in Ref.~\cite[Lemma~18]{brandao_local_2016}:
Consider the operator 
\begin{equation}
\tilde{B}=\sum_{i\in\{0,1\}}|i\rangle\langle\psi_i|^{\otimes n-l},
\end{equation}
with $|i\rangle$ denoting any orthonormal basis of the space $\mathcal{G}_{n-l}$. 
We obtain
\begin{align}
\left|\left|\sum_{i\in\{0,1\}}(|\psi_i\rangle\langle\psi_i)^{\otimes (n-l)}\otimes Y_i\right|\right|_{\infty}&=\left|\left|\sum_{i\in\{0,1\}}\tilde{B}^{\dagger}|i\rangle\langle i|\tilde{B}\otimes Y_i \right|\right|_{\infty}\\
&=\left|\left|\tilde{B}^{\dagger}\left(\sum_{i\in\{0,1\}}|i\rangle\langle i|\otimes Y_i\right)\tilde{B} \right|\right|_{\infty}\\
&\leq \left|\left|\tilde{B}\tilde{B}^{\dagger}\right|\right|_{\infty}\left|\left|\sum_{i\in\{0,1\}}|i\rangle\langle i|\otimes Y_i\right|\right|_{\infty}\\
&\leq \left(1+\frac{\sqrt{2}}{2^{n-l}}\right)\max_i||Y_i||_{\infty}.
\end{align} 
The first inequality follows from the submultiplicativity of the spectral norm and the fact that $||\tilde{B}||^2_{\infty}$ is the maximal singular value of $\tilde{B}$ squared and so is $||\tilde{B}\tilde{B}^{\dagger}||_{\infty}$.
The second inequality follows from~\eqref{eq:approxprojectors}.
\end{proof}

From Lemma~\ref{lemma:Nachtergaele} we obtained an $\alpha>0$ such that $\Delta(H_n)\geq \alpha$ for all $n$.
In summary, we obtain an $\varepsilon$-approximate $2$-design from the sequence of bounds
\begin{align}
g\left(v_{n}^{*k},2\right)\leq g(v_{n},2)^k\leq\left(\left(1+\Delta(H_n)\right)^{-\frac12}\right)^k.
\end{align}
We complete the proof of Theorem~\ref{theorem:t-design} using Lemma~\ref{lemma:expandertodesign} and the bound
\begin{align}
&\left(\left(1+\frac{\Delta(H_n)}{9}\right)^{-\frac12}\right)^k\leq 2^{-4n}\varepsilon\\
\implies& -\frac{k}{2}\log\left(1+\frac{\Delta(H_n)}{9}\right)\geq 4n+\log\left(\frac{1}{\varepsilon}\right)\\
\implies& k\geq\frac{2\left(4n+\log\left(\frac{1}{\varepsilon}\right)\right)}{\log\left(1+\frac{\Delta(H_n)}{9}\right)}\geq \frac{18\left(4n+\log\left(\frac{1}{\varepsilon}\right)\right)}{\Delta(H_n)}\label{eq:correctinequalityk}.
\end{align}
As we group in layers of three to obtain the probability distribution $v_n$, it suffices to choose the depth $D$ of the effective circuit~\eqref{eq:circuitpicture} as
\begin{equation}
D=m-1\geq \frac{54\left(4n+\log\left(\frac{1}{\varepsilon}\right)\right)}{\Delta(H_n)}\in\mathcal{O}\left(4n+\log\left(\frac{1}{\varepsilon}\right)\right),
\end{equation}
which shows that the approximate $2$-design property follows in linear depth.

\section{Numerical results for the gap $\Delta\left(H^B_{7}\right)$}\label{appendix:numericalgap}
As was shown in Appendix~\ref{appendix:proof-2-design}, the spectral gap $\Delta(H_n)$ determines the convergence rate of the effective circuits~\eqref{eq:circuitpicture} to a $2$-design.
We have further proven that the spectral gap is lower bounded by $\Delta\left(H^B_{7}\right)/32$.
This immediately implies the existence of gap in the thermodynamic limit but does not give us the size of this gap.
We can use standard techniques for exact diagonalization to approximate the gap $\Delta\left(H^B_{7}\right)$.
In more detail, the Hamiltonians $H^B_n$ are unitarily equivalent to 
\begin{equation}
\tilde{H}^B_{n} \coloneqq \sum_{i=1}^{n-1}\left(\left(\id-P_i^{XZ}\right)+\left(\id-P_i^{ZX}\right)\right)+\sum_{i=1}^n\left(\id-P_i^Z\right),
\end{equation}
with $P_i^Z$ as in~\eqref{eq:def-projectors} and
\begin{equation}\label{eq:alternativeprojectors}
P^{XZ}_i \coloneqq \frac{1}{2\pi}\int \left(e^{\mathrm{i}\varphi^{XZ}_i X_i\otimes Z_{i+1}}\right)^{\otimes 2,2}\mathrm{d}\varphi_i^{XZ}\qquad P^{ZX}_i \coloneqq \frac{1}{2\pi}\int \left(e^{\mathrm{i}\varphi^{ZX}_i Z_i\otimes X_{i+1}}\right)^{\otimes 2,2}\mathrm{d}\varphi_i^{ZX}.
\end{equation}
The unitary that maps the Hamiltonians is defined as
\begin{equation}
V_n=\begin{cases}\prod_{i=1}^{ n-1}CZ_{2i,2i+1}\quad \text{for}\quad n\;\;\text{even}\\\prod_{i=1}^{ n-2}CZ_{2i,2i+1}\quad \text{for}\quad n\;\;\text{odd}.\end{cases}
\end{equation}
It can be straightforwardly shown that $V_nH^B_nV_n^{\dagger}=\tilde{H}^{B}_{n}$.

The local terms are now $256\times 256$ matrices, which can be computed by \textit{Mathematica} on a laptop using the definition~\eqref{eq:alternativeprojectors}. 
Furthermore, we have used the in-build \textit{python}-function scipy.sparse.linalg.eigsh to compute the first four eigenvalues of $\tilde{H}^B_{7}$.
Notice that the ground space degeneracy is $3$ according to Lemma~\ref{lemma:groundspace}.
For the spectral gap we obtain the following numerical values for small $n$:
	 \begin{center}
		\begin{tabular}{|c|c|c|c|c|c|c|c|c|c|}\hline
			$n$ & \,$3$\,&\,$4$\,&\,$5$\,&\,$6$\,&\,$7$\\ \hline
			$\Delta\left(H_n^B\right)\approx$ & $0.070$&$0.096$&$0.105$&$0.109$&$0.111$\\ \hline
			$\Delta\left(H_n\right)\approx$&$0.241$&$0.202$&$0.175$&$0.157$&$0.146$\\ \hline
		\end{tabular}
    \end{center}
We thus find strong numerical evidence for the size of the constants involved in our asymptotical result.
Plugging this in the Nachtergaele bound yields that $\{ U_\beta^m\}_\beta$ forms an $\varepsilon$-approximate $2$-design in depth $m-1 \approx 15700(4n+\log(1/\varepsilon))$.
If we simply extrapolate the numerics in Appendix~\ref{appendix:numericalgap} we get the more moderate constant $m-1\approx 490(4n+\log(1/\varepsilon))$ as a lower bound sufficient for the $2$-design property.
As these rigorous arguments provide upper bounds, we expect the real constants that yield anticoncentration to be much lower.
\section{Proof of average-case hardness}\label{appendix:proof_average}

This appendix applies the analysis from Ref.~\cite{bouland_quantum_2018} to the architectures developed 
in Ref.~\cite{bermejo-vega_architectures_2018}. 
We obtain average-case hardness results for these architectures.
Drawing an angle $\varphi\in S^1$ uniformly at random means to draw from the Haar measure $H\left(S^1\right)$ of $S^1$ viewed as a compact Lie group.
Furthermore, the product group $(S^1)^{\times Dn}$ is what we draw from in the architecture in Definition~\ref{def:architecture}. We will abbreviate the corresponding Haar measure with
\begin{equation}
H \coloneqq H\left(\left(S^1\right)^{\times Dn}\right).
\end{equation}

Now, we introduce distributions close to the Haar measure as in Ref.~\cite{bouland_quantum_2018}.
\begin{definition}[$\theta$-perturbed Haar-distribution]
	The distribution $H_{\theta}$ is defined by drawing $\{R^l_j\}_{j\in V, l\in [D]}$ from $H$ and then setting the random unitaries $G^l_j \coloneqq R^l_je^{-\ii h^l_j\theta}$ with $r^l_j \coloneqq -\ii\log R^l_j$.
\end{definition}
The intuition for the above definition is to draw from the Haar measure and then `rotate back by a small angle $\theta$`. 
Notice that this is a modification of the definition in Ref.~\cite[Definition~9]{bouland_quantum_2018} where the latter applies to the Haar measure on the group of all unitaries. 
For technical reasons explained further in the proof of Theorem~\ref{theorem:main}, we need to consider a truncated version of the above distribution. 
\begin{definition}[$(\theta,K)$-truncated perturbed Haar distribution]
	The distribution $H_{\theta,K}$ is defined by drawing $\{R^l_j\}_{j\in V,l\in [D]}$ from $H$ and then setting the unitaries
	\begin{equation}
	G^l_j \coloneqq R^l_j\left(\sum_{k=0}^{K}\frac{(-\ii h^l_j\theta)^k}{k!}\right).
	\end{equation}
\end{definition}
We can restrict our analysis to the computation of $p_0(C) \coloneqq |\bra {0^n} C \ket {+^n}|^2$. 
The reason for this is that our architecture admits the \emph{hiding property}, i.e. we can hide the fixed outcome in the circuit without changing the measure. 
Indeed, for a string $y$, we can equivalently view this outcome as $0^n$ for the circuit that has additional $X$ gates inserted before measurement wherever $y$ contains a $1$. 
These $X$ gates are equivalent to a $Z$-gate applied after (or before) the random gate $G_i$ but since $Z$ is diagonal, this does not change the probabilities of the gates $C_i$ because of the translation-invariance of the Haar measure.
We can now state the main theorem:

\begin{theorem}[Main theorem]\label{theorem:main}
	For the input $C$, a circuit as in Definition~\ref{def:architecture}, it is \#P-hard to compute $\frac34+\frac{1}{\poly(n)}$ of the probabilities $p_0(C')$ over the choice $C'$ drawn from $C*H_{\theta,K}$ and over the choice of $\theta\in[0, \poly^{-1}(n))$ with $K=\poly(n)$. Here, the multiplication is meant gate-wise over the random gates.
\end{theorem}

Before we come to the proof, let us briefly discuss what this means. Clearly, for every circuit $C$ it is $C*H=H$ by definition of the Haar measure.
We will show a worst-to-average reduction for the $(\theta,K)$-truncated perturbed Haar measure for 
polynomially large $K$ and small $\theta$.
It therefore \emph{average-case hard} to approximate the outcome probabilities of our circuit over a 
distribution that is indistinguishable from the Haar measure.
For the proof of Theorem~\ref{theorem:main}, we need two lemmas. The first one has been
proven in Ref.~\cite{bermejo-vega_architectures_2018}:

\begin{lemma}[Hardness of approximation in worst case]\label{lemma:worstcasehardness}Approximating the output probabilities of the architecture to within precision $2^{-\poly(n)}$ is \#P-hard.
\end{lemma}

The second lemma is analogous to Ref.~\cite[Fact~14]{bouland_quantum_2018}.

\begin{lemma}[{Analog of Ref.~\cite[Fact~14]{bouland_quantum_2018}}]\label{lemma:fact14}
	Given a circuit $C$ with random gates $\{G^l_j\}_{j,l}$, 
	we can choose independent Haar-random gate entries $\{R^l_j\}_{j,l}$ and use this choice to publish $C_1$ from $C*H_{\theta}$ and $C_2$ from $C*H_{\theta,K}$.
	Then, $|p_0(C_1)-p_0(C_2)|\leq 2^{-\poly(n)}$ for a sufficiently large choice $K=\poly(n)$.
\end{lemma}

The proof is based on the standard bound on truncated Taylor expansion.
\begin{proof}
	Both circuits $C_s$ with $s\in\{1,2\}$ admit the form 
	\begin{equation}
	C_s=\left(\prod_{i\in V} H_i\prod_{i\in V} G^1_{s,i}\prod_{e\in E} CZ_e\right)...\left(\prod_{i\in V} H_i\prod_{i\in V} G^D_{s,i}\prod_{e\in E} CZ_e\right),
	\end{equation}
	with 
	\begin{equation}
	G^l_{1,i}=G^l_iR^l_i\left(\sum_{k=0}^{\infty}\frac{(-\ii r^l_i\theta)^k}{k!}\right),\;\;
	G^l_{2,i} \coloneqq G^l_iR^l_i\left(\sum_{k=0}^{K}\frac{(-\ii r^l_i\theta)^k}{k!}\right).
	\end{equation}
	Thus, the standard (Suzuki) bound on the truncated expansion of the matrix exponential can be applied directly: 
	\begin{equation}
	\left|\left\langle \psi \right| G^l_{1,i}-G^l_{2,i}\left| \phi\right\rangle\right|\leq \frac{\kappa}{K!}
	\end{equation}
	for a constant $\kappa$.
	In particular, the fact that the above expression does not grow with $n$ is an artefact of the fact that all local gates have uniformly bounded norm.
	We use the well known \emph{Feynman path integral representation} of amplitutes:
	\begin{align*}\label{eq:Feynman}
	\bra {0^n} C_s \ket {+^n}=\sum_{y_2^1,...y^D_{n+1}\in\{0,1\}^n}~~\prod_{l=1}^D\left(\left\langle y^{l+1}_{1}\left|E\right|y^l_{n+1}\right\rangle\prod_{i=1}^n\left\langle y_{i+1}^l\left|G^l_{s,i}\right|y^l_{i}\right\rangle\right)
	\end{align*}
	with $y^1_1=+^n$ and $y^{D+1}_{1}=0^n$.
	Applying the standard bound from above to the Feynman expansion of $|\bra {0^n} C_1\ket{+^n}-\bra {0^n} C_2\ket{+^n}|$ and using the triangle inequality we immediately obtain the bound
	\begin{align*}
	|\bra {0^n} C_1\ket{+^n}-\bra {0^n} C_2\ket{+^n}|\leq\frac{2^{\mathcal{O}(\poly(n))}}{K!}.
	\end{align*}
	The claimed bound now follows from the choice $K=\poly'(n)$.
\end{proof}
The last ingredient we need for the proof is the \emph{Berlekamp-Welch-Algorithm}:
\begin{theorem}[Berlekamp-Welch-Algorithm]\label{theorem:berlekamp}
	Let $q$ be a degree $d$ polynomial in a single variable over a field $\mathbb{F}$.
	Suppose we are given $k$ pairs $\{(x_i,y_i)\}$ of elements in $\mathbb{F}$ with all $x_i$ distinct and with a promise that $y_i=q(x_i)$ for at least $\mathrm{max}(d+1,(k+d)/2)$ points. 
	Then, one can recover $q$ exactly in $\poly(k,d)$ deterministic time.
\end{theorem}
\begin{proof}[Proof of Theorem~\ref{theorem:main}]
	Let $C$ be a circuit in our architecture with random gates $G^l_i$. We can draw a family of gate entries $R^l_i$ from $H^G$ and then use it to publish a circuit from $H_{\theta,K}$ with random gate entries 
	\begin{equation}
	R'^{l}_i(\theta) \coloneqq R^l_i\left(\sum_{i=0}^K\frac{(-\ii r_i\theta)^k}{k!}\right)
	\end{equation}
	with $r^l_i=-\ii \log R^l_i$. Furthermore, we publish a circuit $C'$ from $H_{\theta,K}$ by defining random gates $G'^l_i(\theta) \coloneqq G^l_i*R'^l_i(\theta)$. 
	We have to check that the function $q$ defined as $q(\theta) \coloneqq p_0(C'(\theta))$ is a polynomial. Again, this follows from a Feynman path integral representation:
	\begin{equation}
		\bra {0^n} C_s \ket {+^n}=\sum_{y_2^1,...y^D_{n+1}\in\{0,1\}^n}~~\prod_{l=1}^D\left(\left\langle y^{l+1}_{1}\left|E\right|y^l_{n+1}\right\rangle\prod_{i=1}^n\left\langle y_{i+1}^l\left|G'^l_{i}(\theta)\right|y^l_{i}\right\rangle\right)
	\end{equation}
	where each $\left\langle {y^l_{i+1}}\left|G'^l_{i}(\theta)\right| y^l_{i}\right\rangle$ is a polynomial in $\theta$ of degree $K$. 
	Thus, the above expression is a polynomial of degree $nKD$ and the absolute square is a polynomial of degree $2nKD$. 
	We will now prove that computing $\frac34+\frac{1}{\poly n}$ of the above probabilities over the choice of truncated-perturbed Haar-random circuits is enough to compute the probability for $C$ with only a polynomial overhead. However, we know the latter to be \#P-hard from Lemma~\ref{lemma:worstcasehardness}.

	Assume there is a Turing machine $\mathcal{O}$ that can accurately compute $p_0(C')$ for $\frac34+\frac{1}{\poly n}$ of the choices of $C'$.
	Given such a machine, we can use it to define another machine $\mathcal{O}$ that does the following:
	For fixed $\theta_1,...,\theta_k\in \left[0,\poly^{-1}(n)\right)$ it queries $\mathcal{O}(\theta_l)$ for $l=1,...,k$.
	Then, $\mathcal{O}'$ applies the \emph{Berlekamp-Welch Algorithm} to the points $\{(\theta_l,\mathcal{O}(\theta_l))\}$ to compute a polynomial $\tilde{q}$ of degree $2nK$ and then evaluates $\tilde{q}(1)$.
	Using Markov's inequality and the union bound, we obtain that with a probability of at least $\frac12+\frac{1}{\poly n}$ the requirements of the Berlekamp-Welch algorithm are met. 
	This implies $\tilde{q}=q$.
	
	Finally, we have to show that $p_0(C'(1))$ is a $2^{-\poly(n)}$ additive approximation to $p_0(C)$. However, this is the content of Lemma~\ref{lemma:fact14}.
\end{proof}

The careful reader might observe that the distribution $H_{\theta,K}$ takes values outside the unitary group.
However, the gates are exponentially close to unitaries and Theorem~\ref{theorem:main} is necessarily true if the full approximate average-case hardness conjecture holds~\cite{bouland_quantum_2018}.
Moreover, ~\citet{movassagh_efficient_2018} have provided an alternative interpolation, based on the QR-decomposition instead of the geodesics approach taken here in analogy to Ref.~\cite{bouland_quantum_2018}.
This interpolation is in the unitary group such that the 
output probabilities are rational functions with polynomially many parameters to fix. 
Then, one can apply a version of the Berlekamp-Welch algorithm for rational functions~\cite[Alg.~2]{movassagh_efficient_2018}.
 We remark that the same fix carries over to our setting.
 Here, the interpolation based on the QR-decomposition yields a rational interpolation in the subgroup of the gates $\exp(\varphi \ii Z)$.
One can further relax the assumption of exactness to additive errors of the form $2^{-\poly(N)}$, which is crucial to formalize the result in a Turing machine model which has only finite precision. 
This uses powerful results by~\citet{rakhmanov_interpolation_2007} and~\citet{paturi_extrapolation_1992} for the stable interpolation and extrapolation of polynomials.
The corresponding fraction of hard instances is reduced from $1/4 - 1/\poly(n)$ to $1/\poly(n)$, though.
We omit the details here and instead refer to Refs.~\cite{aaronson_bosonsampling_2010,bouland_quantum_2018,haferkamp_average_2018}.

\section{Average-case hardness for generalized circuit sampling and commuting quantum circuits}\label{appendix:generalized-average}
In this appendix we generalize the average-case hardness result in Ref.~\cite{bouland_quantum_2018} and in Appendix~\ref{appendix:proof_average}.
First, we generalize the notion of an \textit{architecture} in Ref.~\cite{bouland_quantum_2018}.

\begin{definition}[Generalized circuit sampling]
 A \emph{generalized architecture} $\mathcal{A}$ is a family of directed graphs $G=(E,V)$, one for each $n$ with equally many adjacent input and output edges at each vertex. 
 We denote with $m_v$ for $v\in V$ the number of input edges.
 Furthermore, every vertex is equiped with a label which is either a Lie-subgroup $G_v$ of $\mathbb{U}(d^{m_v})$ or a fixed unitary in $\mathbb{U}(d^{m_v})$.
 This graph specifies a protocol applied to an input product state.
 If the label is a fixed unitary, we apply this unitary and if it is a subgroup, we draw a unitary Haar-randomly from this subgroup.
 We refer to the gates drawn from Lie-subgroubs as \emph{random gates}.
\end{definition}
Similar to the discussion in Appendix~\ref{appendix:proof_average}, we define the perturbed Haar measure.
\begin{definition}[Perturbed Haar measure]
	The distribution $H^{\mathcal{A}}_{\theta,K}$ is defined by drawing a circuit $A$ with the  generalized architecture $\mathcal{A}$ and local random gates $R_v$ drawn from the Haar measures on $G_v\subseteq \mathbb{U}(d^{m_v})$ and then setting
	\begin{equation}
	G_v=R_v\left(\sum_{k=0}^K\frac{(-\ii\theta r_v)^k}{k!}\right),
	\end{equation}
	with $r_v \coloneqq -\ii\log R_v$.
\end{definition}
This allows us to formulate the following theorem:
\begin{theorem}[Average case hardness in generalized architectures]\label{theorem:generalizedarchitecture}
	Let $\mathcal{A}$ be such that it is \#P-hard to compute the probabilities $|\langle 0^n|U|0^n\rangle|^2$ in the worst-case.
	For a circuit $C$ with underlying architecture $\mathcal{A}$, it is \#P-hard to compute $\frac34+\frac{1}{\poly(n)}$ of the probabilities $p_0(C')$ over the choice $C'$ drawn from $C*H^{\mathcal{A}}$ and over the choice $\theta\in[0,\poly^{-1}(n))$ and $K=\poly(n)$.
\end{theorem}
The proof in Appendix~\ref{appendix:proof_average} can be straightforwardly generalized to show Theorem~\ref{theorem:generalizedarchitecture}.
Furthermore, the result also yields average-case hardness of exact evaluation for IQP circuits~\cite{Bremner} so long as the random angles are chosen uniformly in $[0,2\pi)$ rather than discretely as in the original model.

\end{document}